\documentclass[12pt]{amsart}
\usepackage{amssymb,amsmath,amsthm}
\usepackage{multirow}
\usepackage{enumerate}
\usepackage{fullpage}

\usepackage{hyperref}

\makeatletter\@addtoreset{equation}{section} \makeatother

\renewcommand{\tilde}{\widetilde}

\newcommand{\nad}[2]{\genfrac{}{}{0pt}{}{#1}{#2}}
\newcommand{\p}{\partial}

\newcommand{\mref}[1]{(\ref{#1})}
\renewcommand{\a}{\alpha}
\def\beq#1#2\eeq{%
        \begin{equation}%
        \label{#1}%
            #2%
        \end{equation}%
    }

\def \a {\alpha}
\def \vf {\varphi}
\def \A {\mathcal{A}}

\def \C {\mathbb{C}}
\def \R {\mathbb{R}}
\def \Z {\mathbb{Z}}
\def \N {\mathbb{N}}
\def \amn {\mathcal{A}_{(m,1^n)}}

\newcommand{\arrange}{ \mathcal{A}_{(m,1^n)} }
\newcommand{\rank}{\text{rk}}

\newtheorem{theorem}[equation]{Theorem}
\newtheorem{proposition}[equation]{Proposition}
\newtheorem{lemma}[equation]{Lemma}
\newtheorem{corollary}[equation]{Corollary}

\theoremstyle{remark}
\newtheorem{remark}[equation]{Remark}

\theoremstyle{definition}

\newtheorem{definition}[equation]{Definition}

\title{A class of  Baker-Akhiezer arrangements}
\author{M.\,Feigin, D.\,Johnston}

\begin{document}

\begin{abstract}
We study a class of arrangements of lines with multiplicities on the plane which admit the Chalykh-Veselov Baker-Akhiezer function. These arrangements are obtained by adding multiplicity one lines in an invariant way to any dihedral arrangement with invariant multiplicities.
We describe all the Baker-Akhiezer arrangements  when at most one line has multiplicity higher than 1. We study associated algebras of quasi-invariants which are isomorphic to the commutative algebras of quantum integrals for the generalized Calogero-Moser operators. We compute the Hilbert series of these algebras and we conclude that the algebras are Gorenstein. We also show that there are no  other arrangements with Gorenstein algebras of quasi-invariants when at most one line has multiplicity bigger than 1.
\end{abstract}

\address{School of Mathematics and Statistics, University of Glasgow, 15 University Gardens,
Glasgow G12 8QW, UK}

\email{misha.feigin@glasgow.ac.uk, davidjohnston12@hotmail.co.uk}

\maketitle

\section{Introduction}

The notion of a multi-dimensional Baker-Akhiezer (BA)  function 
was introduced by Chalykh and Veselov in \cite{CV}
in connection with quantum Calogero-Moser systems. In that context the BA function is a special common eigenfunction of the Calogero-Moser operator and its quantum integrals.
Chalykh, Styrkas, Veselov and one of the authors studied the BA functions associated
with finite sets of  vectors in $\mathbb{C}^N$ taken with integer multiplicities \cite{CSV}, \cite{CFV}. Besides the  relevance to quantum integrable systems of Calogero-Moser type it was established in \cite{CFV} that the BA functions are closely related to the Huygens' Principle in the Hadamard sense (see also \cite{BerVes}, \cite{BL}).

Let us recall the construction of \cite{CSV} in more detail. Let $A$ be a finite set of non-collinear vectors $\alpha_i \in \mathbb{C}^N$, $ 0 \leq i \leq n$, for some $n \in \N$. Let $m: A\rightarrow \mathbb{N}$ be the multiplicity function. Let $m_i= m(\a_i)$. For $\lambda=(\lambda_1, \ldots ,\lambda_N)$ and $x=(x_1, \ldots, x_N)$ let $(\lambda,x)=\displaystyle \sum_{i=1}^N \lambda_ix_i$ be the standard inner product in $\mathbb{C}^N$.
\begin{definition}{\label a}
A function $\phi(\lambda, x), \lambda, x\in \mathbb{C}^N$ is called the {\it Baker-Akhiezer function} associated with the configuration
 $\mathcal{A}=(A,m)$
 if the following two conditions are fulfilled:
\begin{itemize}
\item $\phi(\lambda,x)$ has the form $$\phi(\lambda,x)=P(\lambda,x)e^{(\lambda,x)},$$ where $P(\lambda,x)$ is a polynomial in $\lambda$ and $x$ with the highest term $ \prod_{\alpha_i \in A}(\alpha_i, \lambda)^{m_i}     (\alpha_i, x)^{m_i}$;
\item For all $\alpha_i \in A$
 $$\partial_{\alpha_i}^{2s-1}\phi(\lambda,x)|_{\Pi_i}=0,$$ where $1 \le s \le m_i$,
$\Pi_i$ is  the hyperplane $(\alpha_i, x)=0$, and $\partial_{\alpha_i}=(\alpha_i, \frac{\partial}{\partial x})$ is the normal derivative for this hyperplane.
\end{itemize}
\end{definition}
The following conditions provide an effective way to check if the BA function exists.

\begin{theorem}\cite{CSV, CFV}\label{thm1} The Baker-Akhiezer function exists for a configuration $\mathcal{A}$ if and only if  for any $\alpha_j \in A$ the following two sets of conditions hold
\begin{equation}\label{1stconds}
\sum_{i=0, \, i\neq j}^n\frac{m_{i}(\alpha_j, \alpha_i)^{2k-1}}{(x, \alpha_i)^{2k-1}}|_{\Pi_j}=0,   \tag{$\alpha_j(k)$}
\end{equation}
\begin{equation}\label{2ndconds}
\sum_{i=0, \, i\neq j}^n\frac{m_i(m_{i}+1)(\alpha_i, \alpha_i)(\alpha_j, \alpha_i)^{2k-1}}{(\alpha_i, x)^{2k+1}}|_{\Pi_j}=0,   \tag{$\tilde{\alpha}_j(k)$}
\end{equation}
where $1 \leq k \leq m_{j}$.
\end{theorem}

The BA function exists for very special configurations only. Chalykh, Styrkas and Veselov showed that it exists in the Coxeter case that is for $A=R_+$ a positive half of a Coxeter root system $R$ and for the multiplicity function being invariant under the corresponding Coxeter group $W$  \cite{CSV} (see also \cite{CV}). Furthemore if all the multiplicities are equal then there are no other cases \cite{CSV}. The only non-Coxeter examples for which $\phi$ is known to exist are the deformed root systems $\mathcal{A}_N(m)$ and $\mathcal{C}_{N}(m,l)$, which were found by Chalykh, Veselov and one of the authors \cite{CFVdeformednote, CFVnote2}. These configurations can be viewed as deformations of the positive halves of the classical root systems $A_N$ and $C_N$.

In the paper \cite{CFV} Chalykh, Veselov and one of the authors introduced function $\psi$ by modifying the  axiomatics of the Baker-Akhiezer function. The idea was to make the axiomatics less restrictive in order to obtain a richer class of algebraically integrable Schr\"odinger operators of Calogero-Moser type. It was shown that if the BA function $\phi$ exists then so too does $\psi$ and $\phi=\psi$, although the converse is not true. The equations (\ref{2ndconds}) are both necessary and sufficient for the existence of $\psi$ \cite{CFV}.

A direct description of all configurations admitting the BA functions or those satisfying weaker restrictions (\ref{2ndconds}) (so called \emph{locus configurations}) is not obtained yet in general. On the plane a description of all the locus configurations follows with the use of  \cite{CFV} from the results of Berest and Lutsenko \cite{BL}, \cite{B} who found all two-dimensional operators which solve a restricted Hadamard's problem. Muller  investigated the existence of planar locus configurations in more detail  and showed that for any  `coarsely symmetric'  collection of multiplicities there exists a unique  real locus configuration \cite{Muller}.

In this paper we  study a class of  planar locus configurations which admit the BA function. We show that all Muller's locus configurations with coarsely symmetric multiplicities admit the BA functions, and we describe these configurations. Conjecturally, this provides all the planar locus configurations admitting the BA functions. We also study the rings of quantum integrals of the associated operators of Calogero-Moser type. We show that these rings have nice algebraic properties that do not hold in case of general locus configurations. Another nice feature of this subclass of planar locus configurations is explored in \cite{FHV} where the integrals of Macdonald-Mehta type associated with these configurations are explicitly computed.

Recall the observation of \cite{CV} that the algebras of quantum integrals for the Calogero-Moser systems become larger at the special values of coupling parameters. In the case of Calogero-Moser system associated to a Coxeter group $W$ this algebra extends the algebra of $W$-invariant polynomials $\C[x_1,\ldots,x_N]^W$ which is isomorphic to the algebra of quantum integrals at generic parameters.

\begin{definition}(cf. \cite{CSV}, \cite{FV}). \label{qinvsdefn} A polynomial $p \in \mathbb{C}[x_1, x_2, \ldots ,x_N]$ is called \emph{quasi-invariant} with respect to $\mathcal{A}=(A, m)$ if for all $\alpha_i \in A$ and $1\le s \le m_i$ one has
$\partial_{\alpha_i}^{2s-1}p|_{\Pi_i}=0$.
\end{definition}

These polynomials form an algebra $Q_{\mathcal{A}}$.
The relation of the algebra $Q_{\mathcal{A}}$ to the BA function $\phi$ and integrability of the corresponding generalized Calogero-Moser system is clarified by the following result.
\begin{theorem}\cite{CV}, \cite{CSV} If the Baker-Akhiezer function $\phi(\lambda,x)$ exists then for any polynomial $f \in Q_{\mathcal{A}}$ there exists a differential operator $L_f=L_f(x, \frac{\partial}{\partial x})$ such that
\begin{equation*}L_f\phi(\lambda,x)=f(\lambda)\phi(\lambda,x).
\end{equation*}
These operators form a commutative algebra isomorphic to $Q_{\mathcal{A}}$. The operator
\begin{equation}\label{CM-operator}
L=\Delta-\sum_{\alpha \in A}\frac{m_{\alpha}(m_{\alpha}+1)(\alpha,\alpha)}{(\alpha,x)^2}
\end{equation}
corresponds to $f(\lambda)=\lambda^2$.
\end{theorem}

In the Coxeter case  the algebra of quasi-invariants  $Q_{\mathcal{R}}$ is isomorphic to the maximal commutative algebra containing all $W$-invariant quantum integrals $L_f$ for the corresponding generalized Calogero-Moser operator \eqref{CM-operator} \cite{FV}. The algebras $Q_{\mathcal{R}}$ also have good algebraic structure, namely they are Gorenstein.
In the two-dimensional case this was established by Veselov and one of the authors in \cite{FV} where the Hilbert series of $Q_{\mathcal R}$ for any dihedral configuration $\mathcal R$ with constant multiplicity function was computed. For the general Coxeter case the Gorenstein property of  $Q_{\mathcal R}$
was established by Etingof and Ginzburg in \cite{eg}. Using results of Felder and Veselov \cite{FeV}, the  work  \cite{eg} also led to the calculation of the Hilbert series of $Q_{\mathcal R}$.

Although the  algebras  $Q_{\mathcal{A}}$ do not in general have the nice algebraic properties evident in the Coxeter case there exist interesting examples for which these properties do appear. In \cite{FeiginVes} it was shown that the rings $Q_{\mathcal{A}}$ corresponding to the configurations $\mathcal{A}_N(m)$ and $\mathcal{C}_{N}(m,l)$ are  Cohen-Macaulay. Feigin and Veselov also calculated the Hilbert series for the quasi-invariants corresponding to the two-dimensional configurations  $\mathcal{A}_2(m)$ and $\mathcal{C}_{2}(m,l)$ and found that the corresponding algebras are  Gorenstein.

In the present paper we compute the Hilbert series of the algebras of quasi-invariants for all the two-dimensional configurations admitting  the BA functions when at most one multiplicity is bigger than 1. It appears that the algebras are Gorenstein. We are exploring this property further  trying to approach integrability of the  generalized Calogero-Moser systems by studying the algebras of quasi-invariants of the corresponding configurations. It was noted by Etingof and Ginzburg in \cite{eg} that configurations with Gorenstein quasi-invariants are exceptional: in the case of two lines with multiplicities 1 the lines have to be orthogonal.
We describe all configurations on the plane with Gorenstein algebras of quasi-invariants when at most one multiplicity is bigger than~1. Remarkably, this brings us back to the configurations appearing in the beginning of this paper that is all such configurations admit the BA function.

The structure of the paper is as follows. In Section \ref{section-m1n} we study the class of configurations in $\mathbb{C}^2$ which admit the BA functions and which have type $(m,1^n)$ that is  $n$ vectors have multiplicity 1 and one vector has multiplicity $m\ge 1$. We show that there exists a configuration of this type for any number of vectors $n+1$ and for any multiplicity $m$ and that it is unique and real. This configuration is denoted $\arrange$, it is introduced in Definition \ref{m1n}. We firstly show in Theorem \ref{1st-id-implies-condig} that the conditions \eqref{1stconds} uniquely fix the configuration and then we verify in Proposition \ref{firstidimplieslocus} that the locus conditions \eqref{2ndconds} also hold.

In Section \ref{mm1n} we construct further planar configurations admitting the BA functions. Firstly we introduce a configuration $\A_{(m, \tilde m, 1^n)}$. It consists of pairwise non-collinear vectors $\a_0,\ldots, \a_{n+1}\in \R^2$ with the multiplicities $m_0=m, m_{n+1}=\tilde m$, and $m_i = 1$ for $i=1,\ldots,n$. We show in Theorem \ref{BA-exists-two-arb-mult} that $\A_{(m, \tilde m, 1^n)}$ admits the BA function. Then for any $q\in \N$ we introduce  configurations ${\mathcal{A}^q_{(m,\tilde m, 1^n)}} $ which also admit the BA functions (Corollary \ref{corr-on-qexpansionm1n}). The corresponding lines form dihedral arrangement with $2q$ lines with multiplicities $m, \tilde m$ and $n/2$ dihedral orbits of multiplicity 1.

In Section \ref{DarbouxTr}  we construct Darboux-Crum transformations that relate the Sturm-Liouville operators associated with the configuration $\A^q_{(m, \tilde m, 1^n)}$ with the operators with trivial potential. Using results of Berest, Cramer and Eshmatov \cite{BCE} this leads to the explicit formulas for the Baker-Akhiezer functions for the configurations  $\A^q_{(m, \tilde m, 1^n)}$.

In the remainder of the paper we turn our attention to the algebras of quasi-invariants $Q_{\mathcal{A}}$, where $\mathcal{A}$ has type $(m,1^n)$. We start with the preliminary analysis of the quasi-invariant conditions in Section \ref{quasi-preliminary}, where we find partial information on the Hilbert series.
In Section~\ref{poincare-am1n} we study quasi-invariant conditions in more detail. Most of the analysis is  relevant to the configuration $\amn$. Thus we determine the values of the symmetric polynomials in different coordinates of the configuration vectors, these values fix the configuration $\amn$. Further, in Proposition~\ref{lemma2m+2n-2} we establish that all the configurations of type $(m, 1^n)$ have the same dimension of the space of homogeneous quasi-invariants of degree $2(m+n-1)$ except for the configuration $\amn$. The main result of this Section is Theorem~\ref{anm-is-gor}     which gives the Hilbert series of the algebra of quasi-invariants $Q_{\amn}$. An immediate corollary is the fact that the algebra is  Gorenstein.

The main result of Section \ref{all-gorenstein-am1n} is the statement that there are no other configurations of type $(m, 1^n)$ with Gorenstein quasi-invariants except the configuration $\amn$. This is achieved by further study and analysis of possible Hilbert series and by making further use of Proposition~\ref{lemma2m+2n-2}.

\section{Baker-Akhiezer configurations of type $(m,1^n)$}
\label{section-m1n}

We will work in dimension $N=2$. It is convenient to make use of polar and complex coordinates instead of the standard Euclidean coordinate system. The identities from Theorem \ref{thm1} are rearranged as follows.

\begin{lemma}\label{1point} Consider the non-collinear vectors $\alpha_0, \alpha_1, \alpha_2, \ldots ,\alpha_n \in \mathbb{C}^2$. Let us assume that $(\alpha_s, \alpha_s)=1$ for all $s=0,\ldots, n$, and let $\alpha_s=(\cos{\varphi_s}, \sin{\varphi_s})$ for some $\varphi_s$.  Put $z_{s}=e^{2i\varphi_{s}}$.
Then for $1 \leq k \leq m_j$ the condition \mref{1stconds} is equivalent to
\beq{1stcondspolar}
\sum_{\nad{i=0}{i\neq j}}^n \frac{m_i(z_{i}+z_{j})^{2k-1}}{(z_{i}-z_{j})^{2k-1}}=0, \tag{$z_j(k)$}
\eeq
and the condition \mref{2ndconds} is equivalent to
\beq{2ndcondspolar}
 \sum_{\nad{i=0}{i\neq j}}^n\frac{m_{i}(m_{i}+1)z_{i}(z_{i}+z_{j})^{2k-1}}{(z_{i}-z_{j})^{2k+1}}=0.  \tag{$\tilde{z}_j(k)$}
\eeq
\end{lemma}
\begin{proof} We consider the identity \mref{2ndconds}, the identity \mref{1stconds} is similar. Since $(\alpha_j, x)=0$ we take $x=(-\sin\varphi_{j}, \cos\varphi_{j})$. Then relation \mref{2ndconds} takes the form
\begin{align*}
&\sum_{\nad{i=0}{i\neq j}}^n \frac{m_{i}(m_{i}+1)(\cos\varphi_{i}\cos\varphi_{j}+\sin\varphi_{i}\sin\varphi_{j})^{2k-1}}{(-\cos\varphi_{j}\sin\varphi_{i}+\cos\varphi_{i}\sin\varphi_{j})^{2k+1}}=0,
\end{align*}
or, equivalently,
\begin{align*}
&\sum_{\nad{i=0}{i\neq j}}^nm_{i}(m_{i}+1)\frac{\cos^{2k-1}(\varphi_{j}-\varphi_{i})}{\sin^{2k+1}(\varphi_{j}-\varphi_{i})}=0.
\end{align*}
The latter relation takes the form \mref{2ndcondspolar} upon expressing trigonometric functions through the exponents.
\end{proof}

We say that a configuration $\mathcal A$ {\it has type} $(m,1^n)$ if $\mathcal A$ contains $n+1$ pairwise non-collinear vectors $\alpha_0, \alpha_1,\ldots,\alpha_n\in \C^2$ with the respective multiplicities $m_0=m \in \N$, $m_i=1$ for all $i=1,\ldots,n$. We are going to study  configurations which admit the BA functions. It is easy to see that if $\A$ admits the BA function then so does any other configuration obtained by applying an orthogonal transformation to $\A$. Furthermore, vectors from $\A$ may be multiplied by scalars preserving the existence of the BA function. We call configurations obtained this way {\it equivalent}. We assume throughout the paper that all the vectors in  the configurations we consider are non-isotropic: $(\a_i,\a_i)\ne 0$. First we prove the following.

\begin{proposition}
Let $\A$ be a configuration of non-collinear vectors $\a_0,\a_1,\ldots,\a_n \in \C^2$ with corresponding multiplicities $m_0=m \in \R$, $m_j=1$ for $j=1,\ldots,n$. Let $z_j$ be defined as in Lemma \ref{1point}. Consider the  polynomial
$P(w)= \prod_{j=1}^n (w-z_j)$.  Then $\A$ satisfies the identities $(z_j(1))$ for all $j=1,\ldots, n$ if and only if $P(w)$ satisfies the second order differential equation
\beq{difeq}
w(w-z_0) P'' - ( (n-1)(w-z_0) - m(w+z_0) )P' - mn P=0.
\eeq
\end{proposition}
\begin{proof}
We have $P'(z_1)=\prod_{j=2}^n (z_1-z_j)$ and
$$
\frac{P''(z_1)}{P'(z_1)} = \sum_{j=2}^n \frac {2}{z_1-z_j}.
$$
The relation $(z_1(1))$ has the form
$$
\sum_{j=2}^n \frac{z_1+z_j}{z_1-z_j} + \frac{m(z_1+z_0)}{z_1-z_0}=0,
$$
which can be rearranged as
$$
z_1 \frac{P''(z_1)}{P'(z_1)} + \frac{m(z_1+z_0)}{z_1-z_0}=n-1.
$$
Equivalently,
\beq{diffeqproof}
(w-z_0)w P''(w)+(m(w+z_0)-(n-1)(w-z_0))P'(w)=0
\eeq
if $w=z_1$. Similarly for any $j=1,\ldots,n$ the left-hand side of \mref{diffeqproof} vanishes at $w=z_j$ if and only if $(z_j(1))$ holds. Since the left-hand side of \mref{diffeqproof} is a polynomial in $w$ of degree $n$ the collection of conditions $(z_j(1))$ for $j=1,\ldots,n$ is equivalent to the differential equation \mref{difeq} as stated.
\end{proof}

\begin{theorem}\label{1st-id-implies-condig}
Let $\mathcal A$ be a configuration of type $(m,1^n)$. Let $\varphi_j, z_j$ be same as in Lemma~\ref{1point}, and suppose that the conditions $(z_j(1))$ hold for all $j=0,\ldots,n$. Let $e_k$ be the $k$-th elementary symmetric polynomial in the variables $z_1,\ldots, z_n$. Then
\begin{equation}\label{ehats}
e_k= \frac{(-1)^k \prod_{j=1}^k(m+j-1)(n-j+1)z_0^k}{k!\prod_{j=1}^k(m+n-j)}=\frac{(-1)^k\binom{n}{k}\binom{m+k-1}{k}z_0^k}{{\binom{m+n-1}{k}}}
\end{equation}
for any $k=1,\ldots,n$.
\end{theorem}
\begin{proof}
Let
$$
P(w)= \prod_{k=1}^n (w-z_k) = \sum_{k=0}^n (-1)^{n-k} e_{n-k} w^k,
$$
where $e_0=1$. The polynomial $P(w)$ satisfies the differential equation \mref{difeq}. By considering terms of degree $n-1$ we get
$e_1 = - m n z_0/(m+n-1)$ as required. We also have the following recurrence relation
\beq{rec}
z_0 (k+1)(n+m-k-1) e_{n-k-1} + (n-k)(m+k)e_{n-k}=0
\eeq
for any $k= 0, \ldots, n-2$. These relations uniquely determine all the elementary symmetric polynomials $e_k$, $k=1,\ldots,n$. It is easy to see that the values \mref{ehats} satisfy \mref{rec}. Hence the statement follows.

\end{proof}
Theorem \ref{1st-id-implies-condig} has the following corollary.

\begin{proposition}
There exists at most one up to equivalence configuration of type $(m,1^n)$ which satisfies identities $(z_j(1))$ for all $j=0,\ldots,n$. The corresponding arrangement of lines  is symmetric with respect to the line of multiplicity $m$.
\end{proposition}
\begin{proof}
Let us fix the vector of multiplicity $m$ by requiring $z_0=1$. Then the symmetric combinations $e_k$ of the coordinates $z_j$ of the vectors of multiplicities 1 are all determined by Theorem \ref{1st-id-implies-condig}. This fixes the configuration. Let
$$
P(z)=\prod_{i=1}^n(z-z_i)=\sum_{i=0}^n(-1)^ie_iz^{n-i}=\sum_{i=0}^n(-1)^ie_i z_0^{n-2i} z^{i},
$$
where we used the symmetry
\begin{equation}\label{e}
e_i=(-1)^nz_0^{2i-n}e_{n-i},\,\,\,1\leq i \leq n,
\end{equation}
which can be easily seen from the formulas \mref{ehats}. It follows that if $z=z_i$ is a solution to $P(z)=0$ then $z=z_0^2 z_i^{-1}$ is another solution, so the configuration is symmetric.
\end{proof}
In order to establish existence of the configuration described in Theorem  \ref{1st-id-implies-condig} one has to prove that the corresponding values of elementary symmetric polynomials $e_k$ define a point in $\C^n/S_n$ which is not in the discriminant set. Rather than showing this directly we are going to explain that for any multiplicities $m_0,\ldots,m_n$ there exists a (unique) {\it real} configuration $\mathcal A$ satisfying identities $(z_j(1))$ for all $j=0,\ldots,n$.

\begin{theorem}\label{realexists}
Let $\mathcal A$ be a configuration of vectors $\alpha_0, \ldots, \alpha_n \in \R^2$ with the corresponding multiplicities $m_0,\ldots, m_n \in \R_+$. Let $\varphi_j \in [0,\pi)$ be such that $\a_j = (\cos \varphi_j, \sin \vf_j)$. Suppose that $0=\vf_0 < \vf_1 <\ldots < \vf_n< \pi$.
Then for any choice of the multiplicities $m_0, \ldots,m_n$ there exists a unique configuration
 $\A$ which satisfies the relations $(\a_j(1))$ for all $j=0,\ldots, n$.
\end{theorem}
\begin{proof}
We follow closely the proof of a similar statement established in \cite[Theorem 6.4]{Muller}.  Let $$F(\psi_1,\ldots,\psi_n)= \sum_{\nad{i,j =0}{i <j}}^n m_i m_j  f(\psi_j-\psi_i),$$
where  $f(x)=\log \sin x$ and $\psi_0=0$.
Then the collection of
identities $(\a_j(1))$, $j=0,\ldots,n$,  is equivalent to the property that $F$ has a critical point at $\psi_j = \vf_j$.

One can show that in the region
$0=\psi_0 < \psi_1 <\ldots < \psi_n< \pi$ the function
 $F$ has a unique critical point.  This uses convexity of the function $f(x)$ and the behaviour at the boundaries encoded by $f(x) \to -\infty$, as $x \to 0^+$ or $x \to \pi^-$. This is in full analogy with the case of \cite{Muller} where $f(x) = 1/\sin^2 x$ was considered.
\end{proof}

Now we are ready to introduce a particular configuration $\A$.

\begin{definition}\label{m1n}
 Let $\mathcal{A}$ be a collection of vectors $\alpha_0, \alpha_1, \ldots ,\alpha_n \in \mathbb{C}^2$ with multiplicities $m_{0}=m \in \Z_+$,  $m_{i}=1$ for $1 \leq i \leq n$. Suppose that $\a_0=(1,0)$.
Let $z_i$ be the same as in Lemma \ref{1point}, and  denote by $e_i$ the $i$-th elementary symmetric polynomial in the variables $z_i$, $1\leq i \leq n$. We say that $\mathcal{A}$ is the $\mathcal{A}_{(m,1^n)}$ configuration if
$$e_i=(-1)^i\binom{n}{i}\binom{m+i-1}{i}{\binom{m+n-1}{i}}^{-1}$$ for $1 \leq i \leq n$.
\end{definition}

It follows from Theorems \ref{1st-id-implies-condig}, \ref{realexists}  that the configuration $\mathcal{A}_{(m,1^n)}$ consists of $n+1$ pairwise non-collinear real vectors, and the corresponding arrangement of lines is symmetric with respect to the vector $\a_0$. The latter property can be rephrased as $z_i z_{n-i+1}=1$ for $1 \leq i \leq n$ where the vectors with multiplicity 1 are ordered appropriately.
In the case $m=1$ the vectors $\a_0,\ldots,\a_n$ are the normal vectors of the dihedral arrangement of $n+1$ lines in $\R^2$.

It appears that the configuration $\amn$ admits the BA function.

\begin{theorem}\label{BAanm}
A configuration $\A$ of type $(m,1^n)$ admits the BA function if and only if $\A$ is equivalent to the configuration $\amn$.
\end{theorem}
It follows from Theorem  \ref{1st-id-implies-condig} and its proof that  a configuration $\A$ of type $(m,1^n)$ satisfies the identities $(\a_j(1))$ for all $j=0,\ldots, n$ if and only if $\A$ is equivalent to the configuration $\amn$. Since the configuration $\amn$ is symmetric the identities $(\tilde \alpha_0(k))$, $k=1,\dots,m$ are satisfied. Thus Theorem \ref{BAanm} is reduced to the following statement.

\begin{proposition}\label{firstidimplieslocus}
Suppose a configuration $\A$ of type $(m,1^n)$ satisfies the identities $(\alpha_j(1))$ for all $j=0,\ldots, n$. Then $\A$ satisfies the identities $(\tilde \alpha_j(1))$ for all $j=0,\ldots, n$.
\end{proposition}

Before proving this proposition we analyze the identities $(\tilde z_j(1))$.
\begin{lemma}\label{lemma214}
Suppose a configuration $\A$ of type $(m,1^n)$ satisfies the condition $(z_1(1))$. Then the relation $(\tilde z_1(1))$ can be rearranged as
\begin{align}\label{star}
&\left(2m(m+1)z_0z_1(z_0+z_1)+2(z_1-z_0)^2 \gamma -3(z_1-z_0)\gamma^2+\gamma^3\right) P'(z_1)  \\
&+2z_1^2(z_1-z_0)^2\left( 2(z_1-z_0)-\gamma \right)P^{(3)}(z_1) +z_1^3(z_1-z_0)^3P^{(4)}(z_1)=0,\notag
\end{align}
where $P(z)=\prod_{j=1}^n (z-z_j)$, and
$\gamma =(n-1)(z_1-z_0)-m(z_1+z_0)$.
\end{lemma}
\begin{proof}
Firstly we use the identity
 \begin{equation}\label{locusid}
\sum_{i=2}^{n}\frac{(z_1+z_i)z_i}{(z_1-z_i)^3}
=-3z_1\sum_{i=2}^{n}\frac{1}{(z_1-z_i)^2}+2z_1^2\sum_{i=2}^{n}\frac{1}{(z_1-z_i)^3}+\sum_{i=2}^{n}\frac{1}{z_1-z_i}
\end{equation}
to rearrange  the condition $(\tilde z_1(1))$ into
\beq{z1tilde}
\frac{m(m+1)(z_1+z_0)z_0}{2(z_1-z_0)^3}
+ f(z_1)+3z_1f^{'}(z_1)+z_1^2f^{''}(z_1) =0,
\eeq
where we introduced
$$
f(z)= \sum_{i=2}^n\frac{1}{z-z_i} = \frac{P^{'}(z)}{P(z)} - \frac{1}{z-z_1}.
$$
Note that due to the identity $(z_1(1))$ one has

$$m\frac{z_1+z_0}{z_1-z_0}+\sum_{i=2}^n\frac{z_1+z_i}{z_1-z_i}= m\frac{z_1+z_0}{z_1-z_0}+ 2z_1f(z_1)-(n-1) =0,
$$
hence
\begin{equation}\label{f(z_1)}
f(z_1)=\frac{P''(z_1)}{2P'(z_1)}=\frac{1}{2z_1}\left(n-1-\frac{m(z_1+z_0)}{z_1-z_0}\right) = \frac{\gamma}{2 z_1(z_1-z_0)}.
\end{equation}

Now we express $f^{'}(z_1)$ and  $f^{''}(z_1)$ in terms of $P$ by considering the Taylor expansion of $f(z)$ at $z=z_1$. Put $\varepsilon=z-z_1$. We have
{
\begin{align*}f(z)&=\frac{P^{'}(z)(z-z_1)-P(z)}{P(z)(z-z_1)}
=\frac{\frac{1}{2}P^{(2)}(z_1)+\frac{1}{3}P^{(3)}(z_1)\varepsilon+\frac{1}{8}P^{(4)}(z_1)\varepsilon^2+O(\varepsilon^3)}{P^{'}(z_1)+\frac{1}{2}P^{(2)}(z_1)\varepsilon+\frac{1}{6}P^{(3)}(z_1)\varepsilon^2+O(\varepsilon^3)}\\
&=\frac{1}{P^{'}(z_1)}\left(\frac{1}{2}P^{(2)}(z_1)+\frac{1}{3}P^{(3)}(z_1)\varepsilon+\frac{1}{8}P^{(4)}(z_1)\varepsilon^2+ O(\varepsilon^3)\right) \\
&\times\left(1-\frac{P^{(2)}(z_1)}{2P^{'}(z_1)}\varepsilon-\frac{P^{(3)}(z_1)}{6P^{'}(z_1)}\varepsilon^2+(\frac{P^{(2)}(z_1)}{2P^{'}(z_1)})^2\varepsilon^2+O(\varepsilon^3)\right)\\
&=\frac{P^{(2)}(z_1)}{2P^{'}(z_1)}+\frac{1}{P^{'}(z_1)}\left(\frac{P^{(3)}(z_1)}{3}-\frac{P^{(2)}(z_1)^2}{4P^{'}(z_1)}\right)\varepsilon+\frac{1}{P^{'}(z_1)}\left(\frac{1}{8}P^{(4)}(z_1) \right.
\\& -\frac{P^{(3)}(z_1)P^{(2)}(z_1)}{6P^{'}(z_1)}
\left.+\frac{1}{4}P^{(2)}(z_1)\left(\frac{P^{(2)}(z_1)^2}{2P^{'}(z_1)^2}-\frac{P^{(3)}(z_1)}{3P^{'}(z_1)}\right)\right)\varepsilon^2
+O(\varepsilon^3).\end{align*}
}
Thus we have the identities
\begin{align}\label{f'}
f^{'}(z_1)&=\frac{1}{P^{'}(z_1)}\left(\frac{P^{(3)}(z_1)}{3}-\frac{P^{(2)}(z_1)^2}{4P^{'}(z_1)}\right),\\ \label{f''}
f^{''}(z_1)&=\frac{1}{P^{'}(z_1)}\left(\frac{P^{(4)}(z_1)}{4}-\frac{P^{(3)}(z_1)P^{(2)}(z_1)}{2P^{'}(z_1)}+
\frac{P^{(2)}(z_1)^3}{4 P'(z_1)^2}\right).
\end{align}
Now we substitute expressions \mref{f(z_1)}, \mref{f'}, \mref{f''} into \mref{z1tilde}. We get

\begin{align*}
&{m(m+1)z_0z_1(z_1+z_0)}+ (z_1-z_0)^2 \gamma
+2 z_1^2(z_1-z_0)^3\left(\frac{P^{(3)}(z_1)}{P^{'}(z_1)}- 3 f(z_1)^2\right)\\
&+z_1^3(z_1-z_0)^3\left(\frac{P^{(4)}(z_1)}{2P^{'}(z_1)}-
\frac{2P^{(3)}(z_1)}{P'(z_1)} f(z_1)+4 f(z_1)^3\right)=0.
\end{align*}
By substituting the equality $f(z_1)=  \gamma/(2 z_1(z_1-z_0))$ we arrive at the required relation.
\end{proof}

We are now ready to establish Proposition \ref{firstidimplieslocus}.
\begin{proof}
By differentiating the equation \mref{difeq} we get the following two differential equations for $P(w)$:
\beq{difeq2}
w(w-z_0) P^{(3)} - ( (n-2)(w-z_0) - m(w+z_0) -w)P'' - (mn-m+n-1) P'=0,
\eeq
\beq{difeq3}
w(w-z_0) P^{(4)} - ( (n-3)(w-z_0) - m(w+z_0) - 2 w)P^{(3)} - (mn-2m+2n-4) P''=0.
\eeq
We use equations \mref{difeq}, \mref{difeq2}, \mref{difeq3} to rearrange the left-hand side of \mref{star}. We have
$$
2 z_1^2(z_1-z_0)^2\left( 2(z_1-z_0)-\gamma \right)P^{(3)}(z_1) +z_1^3(z_1-z_0)^3P^{(4)}(z_1)= a P'(z_1) + b P(z_1),
$$
where $a, b$ are polynomials in $z_1, z_0$ with
$$
a=
\left((\gamma -2 z_1 +z_0)\gamma + (m n -m+n-1)z_1(z_1-z_0) \right) \left(-\gamma-2z_0\right)
$$$$+ z_1(z_1-z_0)(mn-2m+2n-4)\gamma.
$$
Since $P(z_1)=0$ and $$a +
\left(2m(m+1)z_0z_1(z_0+z_1)+2(z_1-z_0)^2 \gamma -3(z_1-z_0)\gamma^2+\gamma^3\right)
=0$$ the Proposition follows.
\end{proof}

\section{Planar arrangements with dihedral symmetry}
\label{mm1n}

We start by fixing a particular configuration. It consists of pairwise non-collinear vectors $\a_0,\ldots, \a_{n+1}\in \R^2$ with the multiplicities $m_0=m, m_{n+1}=\tilde m$, and $m_i = 1$ for $i=1,\ldots,n$. We require that the vectors $\a_0$ and $\a_{n+1}$ are orthogonal and that the corresponding arrangement of lines is symmetric with respect to reflection about $\a_0$, so in particular $n$ is even. Let us also require that the conditions $(z_j(1))$ are satisfied for all $j=1,\ldots,n$. Indeed, it follows by the arguments of Theorem \ref{realexists} and \cite{Muller} that there exists a configuration  with these properties, and moreover it is unique up to equivalence. We denote it by  $\A_{(m, \tilde m, 1^n)}$.

Our further arguments are also close to the ones used in Section \ref{section-m1n}. We use the identities \mref{1stcondspolar} to find symmetric combinations of the coordinates determining $\a_i$ $(i=1,\ldots,n)$, and we show that the conditions \mref{1stcondspolar} imply the locus conditions \mref{2ndcondspolar}. This leads to the existence of the BA function for the configuration $\A_{(m, \tilde m, 1^n)}$.

\begin{proposition}
Let $z_j$ be defined for the configuration  $\A_{(m, \tilde m, 1^n)}$ in the same way as in Lemma \ref{1point}. Consider the corresponding polynomial
$P(w)= \prod_{j=1}^n (w-z_j)$.  Then $P(w)$ satisfies the second order differential equation
\begin{align}\label{difeq_two_mult}
w(w^2-z_0^2) P'' - \left((n-1)(w^2-z_0^2) - m(w+z_0)^2-\tilde m (w-z_0)^2 \right)P'
\notag \\
=\left(n(m+\tilde m) w + n(m-\tilde m) z_0\right)P.
\end{align}
Further, the elementary symmetric polynomials $e_k= e_k(z_1,\ldots, z_n)$ satisfy the recurrence relations
$$
(m+\tilde m +k -1)(k- n -1)e_{n-k+1} + (n-2k)(m-\tilde m)z_0 e_{n-k}
$$$$
+(m+\tilde m +n - k -1) (k+1) z_0^2e_{n-k-1} =0
$$
for  $k=1,\ldots,n-1$ with $e_n=z_0^n, e_{n-1}=z_0^{n-1}\frac{(m-\tilde m)n}{n+m+\tilde m-1}$.
\end{proposition}
\begin{proof}
We have $z_{n+1}=-z_0$. Then the conditions $(z_1(j))$ take the form
\beq{poltwomult}
w(w^2-z_0^2)P''-((n-1)(w^2-z_0^2)-m(w+z_0)^2-\tilde m (w-z_0)^2))P' = 0
\eeq
if $w=z_j$ for any $j=1,\ldots,n$.  It follows that the polynomial in the left-hand side of \mref{poltwomult} equals $(\beta z+\gamma) P(z)$ with $\beta=n(m+\tilde m)$ and $\gamma \in \C$. By considering terms of degrees 0 and $n$ we get the equations
$$
e_{n-1} z_0^2(n+m+\tilde m -1) + \gamma e_n =0, \,\,  e_1(n+m+\tilde m -1) + 2(m-\tilde m) n z_0 = \gamma.
$$
Using the symmetry $e_{n-1}=e_1 z_0^{n-2}$, $e_n=z_0^n$ we derive $\gamma = (m-\tilde m)n z_0$ and $e_{n-1}=z_0^{n-1}\frac{(\tilde m - m)n}{n+m+\tilde m -1}$ as required.
The recurrence relation also follows.
\end{proof}

\begin{theorem}\label{BA-exists-two-arb-mult}
There exists the BA function for the configuration $\A_{(m, \tilde m, 1^n)}$.
\end{theorem}
\begin{proof}
In a similar fashion to the proof of Lemma \ref{lemma214}, the condition $(z_1(1))$ takes the form
\begin{align}\label{locus-id-two-mult}
&z_1^2 P^{(4)}(z_1)+4 z_1 (1-z_1 f) P^{(3)}(z_1) + 2z_1 (2 z_1 f - 3) f P^{(2)}(z_1) \\ \notag
&
+\left( 4 f + 2z_0\frac{m(m+1)(z_1+z_0)^4-\tilde m (\tilde m+1)(z_1-z_0)^4}{(z_1^2-z_0^2)^3} \right) P'(z_1)=0,
\end{align}
where
$$
f=\frac{(n-1)(z_1^2-z_0^2)-m(z_1+z_0)^2-\tilde m (z_1-z_0)^2}{2 z_1(z_1^2-z_0^2)}.
$$
On the other hand by differentiating \mref{difeq_two_mult} we get
\begin{align}\label{difeq2_twomult}
&z_1(z_1^2-z_0^2)P^{(3)}(z_1) + \left(3z_1^2-z_0^2-(n-1)(z_1^2-z_0^2)+m(z_1+z_0)^2+\tilde m (z_1-z_0)^2 \right) P^{(2)}(z_1)
\\ \notag
&- \left( (2(n-1)+(n-2)(m+\tilde m))z_1+(n-2)(m-\tilde m)z_0\right) P'(z_1) = n(m+\tilde m)P(z_1),
\end{align}
and
\begin{align}\label{difeq3_twomult}
&z_1(z_1^2-z_0^2)P^{(4)}(z_1)
+ \left(2(3z_1^2-z_0^2)-(n-1)(z_1^2-z_0^2)+m(z_1+z_0)^2+\tilde m (z_1-z_0)^2 \right) P^{(3)}(z_1)
\\ \notag
&+ \left( (10-4n+(4-n)(m+\tilde m))z_1+(4-n)(m-\tilde m)z_0\right) P^{(2)}(z_1)
\\ \notag
&- 2(n-1)(m+\tilde m+1)P'(z_1) =0
\end{align}
Using equations \mref{difeq_two_mult}, \mref{difeq2_twomult}, \mref{difeq3_twomult} we can express each of $P^{(4)}(z_1)$, $P^{(3)}(z_1)$, and $P^{(2)}(z_1)$ as a linear combination of $P'(z_1)$ and $P(z_1)$. One can check that after the substitution of these expressions into the relation \mref{locus-id-two-mult} the term with $P'(z_1)$ cancels, hence the locus condition \mref{locus-id-two-mult} is satisfied. In the same way all the conditions $(z_j(1))$ hold for $j=1,\ldots, n$. Due to the symmetry of the configuration $\A_{(m, \tilde m, 1^n)}$ all the conditions \mref{1stconds}, \mref{2ndconds} are satisfied so the BA function exists by Theorem \ref{thm1}.
\end{proof}

In the special case $\tilde m=0$ or $\tilde m =1$ the configuration $\A_{(m, \tilde m, 1^n)}$ is reduced to the configuration $\amn$ and $\A_{(m, 1^{n+1})}$ respectively.
Now we are going to generate further configurations with Baker-Akhiezer functions. Let $\A$ be a configuration of vectors in $\C^2$ with the multiplicity function $m$. For any vector $\a_i\in \A$ let $\a_i=(\cos \vf_i, \sin \vf_i)$ for some $\vf_i \in \C$. For any positive integer $q$ we define a new configuration $T_q(\A)$. The number of vectors in $\widehat \A = T_q(\A)$ is $q$ times the number of vectors in $\A$. For each vector $\a_i \in \A$ we define new vectors $\a_{i,s} \in \widehat \A$, where $s=1,\ldots, q$. Let $\vf_{i,s}=\vf_i +\pi s/q$. Define $\a_{i,s} = (\cos \vf_{i,s}, \sin \vf_{i,s})$. The multiplicity function on $\widehat \A$ is defined by $m(\a_{i,s})=m(\a_i)$. The importance of the operation $T_q$ is explained by the following proposition.

\begin{proposition}\label{prop-on-qexpansion}
Let $\A$ admit the BA function. Then $\widehat \A=T_q(\A)$ also admits the  BA function for any $q \in \Z_{\ge 1}$.
\end{proposition}
\begin{proof}
Let us consider equalities \mref{1stconds} for the configuration $\A$. They are equivalent to the collection of the equalities
\beq{trigproof}
\p_\vf^{2k-1} \left(v(\vf) - m_j \log \sin (\vf-\vf_j)\right)|_{\vf=\vf_j} =0,
\eeq
where  $v(\vf) =  \log \prod_{\nad{\a_i \in \A}{i \ne j}} \sin(\vf-\vf_i)^{m_i}$ and $k=1,\ldots,m_j$. This implies that
$$
\p_\vf^{2k-1} \left(v(q\vf) - m_j \log \sin (q\vf-\vf_j)\right)|_{\vf=\vf_{j,s}} =0
$$
for $s=1,\ldots,q$.
Notice that
$$
\sin (q \vf - \vf_i) = 2^{q-1} \prod_{s=1}^q \sin(\vf-\vf_{i,s}),
$$
and that
$
\sum_{s=2}^q \cot(\vf - \vf_{j,s})
$
is invariant under the transformation $\vf \to -\vf +2\vf_{j,1}$ (the symmetry about $\vf_{j,1}$). It follows that
$$
\p_\vf^{2k-1} \left( \log \prod_{\nad{t =1}{\a_i \in \A,  i \ne j}}^q \sin(\vf-\vf_{i,t})^{m_i} - m_i \log \sin (\vf-\vf_{j,s})\right)|_{\vf=\vf_{j,s}} =0,
$$
which means that equalities \mref{1stconds} for the configuration $\widehat \A$ hold. Similarly the equalities \mref{2ndconds} hold and hence the statement follows by Theorem \ref{thm1}.
\end{proof}

Fix $q\in \Z_{\ge 1}$ and define ${\mathcal{A}^q_{(m,\tilde m, 1^n)}} =T_q({\mathcal{A}_{(m,\tilde m, 1^n)}})$.
\begin{corollary}\label{corr-on-qexpansionm1n}

The configuration ${\mathcal{A}^q_{(m,\tilde m, 1^n)}} $ admits the BA functions. The collection of lines of the corresponding arrangement forms the set of mirrors of the dihedral group $I_2(2q)$ together with $n/2$  $I_2(2q)$-orbits where all the multiplicities are 1.
%
\end{corollary}

\section{Darboux transformations}\label{DarbouxTr}

In this Section we construct specific Darboux-Crum transformation which relates the Sturm-Liouville operator associated with the configuration $\A^q_{(m, \tilde m, 1^n)}$ with the operator with trivial potential.

Define $Q(w) = P(w) w^{-\frac{n}{2}}$, where $P(w)= \prod_{j=1}^n (w-z_j)$ satisfies the differential equation
\mref{difeq_two_mult}. Then $Q(w)$ satisfies the differential equation
\begin{align}\label{difeq_two_mult_Q}
w(w^2-z_0^2) Q_{ww} + \left((w^2-z_0^2) + (w-z_0)^2\tilde m + (w+z_0)^2 m \right)Q_{w}
\notag
-\frac{n}{2}(m+\tilde m+\frac{n}{2})\frac{w^2-z_0^2}{w}Q=0.
\end{align}

Let now $w=e^{2 i \varphi}$
so that
$$
\p_w=-\frac{i}{2w}\p_\vf, \quad \p^2_w = -\frac{1}{4w^2}\p^2_\vf+\frac{i}{2w^2}\p_\vf.
$$
Put $z_0=1$. Then as a function of $\vf$, $Q$ satisfies the differential equation
\beq{Qvf-diff-eq}
Q_{\vf \vf}+2 (m \cot \vf - \tilde  m \tan \vf) Q_\vf +n(2(m+\tilde m)+n) Q=0.
\eeq
Note also that
$$
Q(\vf) = \left(\prod_{j=1}^n (e^{2 i \vf} - e^{2i \vf_j})\right) e^{-i n \vf} 
= \varepsilon (2i)^n \prod_{j=1}^n \sin (\vf - \vf_j),
$$
where
$\varepsilon = e^{i \sum_{j=1}^n \vf_j}$. 
Define $\tilde Q=Q (\cos \vf)^{\tilde m} (\sin\vf)^m$. It follows from \mref{Qvf-diff-eq} that $\tilde Q$ satisfies the differential equation
\beq{eqQtilde}
L_{m,\tilde m} \tilde Q = (n+m+\tilde m)^2 \tilde Q,
\eeq
where
\beq{lmmt}
L_{m,\tilde m} = - \p^2_\vf + \frac{m(m -1)}{\sin^2 \vf} + \frac{\tilde m(\tilde m-1)}{\cos^2 \vf} .
\eeq

Let us apply Darboux transformation to the operator $L_{m,\tilde m}$ at the level $(n+m+\tilde m)^2$ with the help of function $\tilde Q$. That is we represent  $L_{m,\tilde m}- (n+m+\tilde m)^2 = A^* A$ where $A=\p_\vf-\tilde Q_\vf/\tilde Q$, $A^*=-\p_\vf-\tilde Q_\vf/\tilde Q$  and we define new operator $L=A A^*+(n+m+\tilde m)^2$. Then
\beq{Ldarbaux}
L=-\p^2_\vf + \frac{m(m +1)}{\sin^2 \vf}+  \frac{\tilde m(\tilde m+1)}{\cos^2 \vf} +\sum_{j=1}^n \frac{2}{\sin^2(\vf - \vf_j)}.
\eeq
The operator $L$ has singularities on the lines of the configuration ${\mathcal A}_{{(m,\tilde m, 1^n)}}$ (at $z_0=1$) rotated by $\pi/2$. If $\tilde m=0$ then the singularities are on the lines of the rotated configuration
${\mathcal A}_{{(m, 1^n)}}$ (at $z_0=1$).

\begin{theorem}\label{Darboux-theorem}
Suppose $m \ge  \tilde m$ with $\tilde m, m \in \Z_{\ge 0}$. Let $\chi_j(\vf)=\sin(k_j \vf)$, where $1\le j \le m$, and $k_j=j$ for $j=1,\ldots, m-\tilde m$, $k_{m - \tilde m +j}=m - \tilde m +2j$ for $j=1,\ldots,\tilde m-1$, and $k_{m}=\tilde m+m +n$. Then
\beq{Lform}
L=-\p^2_\vf -2\left(\frac{\p}{\p \vf }\right)^2 \log W[\chi_1,\ldots, \chi_{m}],
\eeq
where $W$ is the Wronskian of the corresponding functions $\chi_1(\vf),\ldots, \chi_{m}(\vf)$, and $L$ is given by \mref{Ldarbaux}. Further to that,
\beq{propnu}
Q(\vf)= \nu (\cos \vf)^{-\frac{\tilde m(\tilde m+1)}{2}}(\sin \vf)^{-\frac{m(m+1)}2} W[\chi_1,\ldots, \chi_{m}] ,
\eeq
where $\nu=2^{-\frac{\tilde m(\tilde m+1)}{2}-\frac{m(m-1)}{2}} (-1)^{\frac{m(m-1)}{2}}\prod_{\nad{p,q=1}{p>q}}^{m} (k_p-k_q)^{-1}$.
\end{theorem}
\begin{proof}
It is discussed in \cite{BL} that the operator \mref{lmmt} can be obtained as a sequence of Darboux transformations with the specified functions $\chi_1,\ldots,\chi_{m-1}$:
$$
L_{m, \tilde m}=-\p^2_\vf -2\left(\frac{\p}{\p \vf }\right)^2 \log W[\chi_1,\ldots, \chi_{m-1}],
$$
so that $W[\chi_1,\ldots, \chi_{m-1}]=\mu (\sin \vf)^{\frac{m(m-1)}{2}} (\cos \vf)^{\frac{\tilde m(\tilde m-1)}{2}}$ for some constant $\mu$.

It follows from the properties of Darboux-Crum transformations that the operator $L$ can be expressed by the formula \mref{Lform} where the function $\chi_{m}= a e^{i(\tilde m+m+n)\vf} + b e^{-i(\tilde m+m+n)\vf}$, and that the function $\tilde Q$ has the form $\tilde Q = c W[\chi_1,\ldots, \chi_{m}]/W[\chi_1,\ldots, \chi_{m-1}]$, where $a, b, c$ are some  constants. It is easy to see  that $\tilde Q = a \psi(i(\tilde m+m+n),\vf) + b  \psi(-i(\tilde m+m+n),\vf)$ where $\psi(k,\vf)$ is the trigonometric Baker-Akhiezer function associated with the root system $BC_1$ with the multiplicity $m - \tilde m$ of the short root and the multiplicity $m -1$ of the long root (unless $m=0$ in which case the statement can be verified directly) \cite{Ch}. It follows from the properties of the Baker-Akhiezer functions and the equation \mref{eqQtilde} that $\tilde Q$ has a pole at $\phi=0$ of order $m -1$ unless $a=-b$. Thus we have to have $a=-b$ so $\chi_{m}$ has the required form. The value of the coefficient $\nu$ is obtained by comparing  $e^{i n \vf}$ terms in both sides of \mref{propnu}.
\end{proof}
The case $m<\tilde m$ can be reduced to the case considered in Theorem \ref{Darboux-theorem} by performing a rotation $\vf \to \vf +\frac{\pi}{2}$.  Further, by replacing the Darboux parameters $k_j$ with $q k_j$, $q \in \N$ the formula \mref{Lform} gives the angular part of the Schr\"odinger operator associated with the configuration ${\mathcal{A}^q_{(m,\tilde m, 1^n)}}$. Now, the Baker-Akhiezer functions of the two-dimensional configurations are expressed in \cite{BCE} in terms of Chebyshev polynomials and Darboux transformations data. Thus Theorem \ref{Darboux-theorem} and \cite[Theorems 2,3]{BCE} provide explicit expressions for the Baker-Akhiezer functions for the configurations  ${\mathcal{A}^q_{(m,\tilde m, 1^n)}}$.

\section{Quasi-invariants: preliminary results}\label{quasi-preliminary}
We now turn our attention to the algebras of quasi-invariants $Q_{\mathcal{A}}$, where $\mathcal{A}$ has type $(m,1^n)$ with $m,n \in \Z_{\ge 1}$. In this Section we get a partial information on the Hilbert series of the  algebra $Q_\A$. In Section~\ref{poincare-am1n} we derive the complete Hilbert series for the algebra of quasi-invariants for the configuration $\amn$ and conclude that the algebra is  Gorenstein. In Section \ref{all-gorenstein-am1n} we prove that there are no other configurations of type $(m,1^n)$ with Gorenstein algebras of quasi-invariants except the configuration $\amn$.

Let us fix some notation. Let $\A$ be a finite set of non-collinear vectors $\beta_0, \beta_1, \ldots ,\beta_n$ with the multiplicities $m_0=m$, $m_i=1$ for $1\le i \le n$. Let $\beta_i=(1,\alpha_i)$ with $\alpha_i \in \mathbb{C}$. We fix $\beta_0=(0,1)$.
Let $Q_\A$ be the corresponding algebra of quasi-invariants (see Definition \ref{qinvsdefn}). It is clear that the algebra is graded. Define its Hilbert series
\begin{equation}\label{poincare2}
P_\A(t)=\sum_{k=0}^{\infty}b_kt^k,
\end{equation}
where $b_k= \dim Q_\A^{(k)}$ and $Q_\A^{(k)}$ denotes the space of homogeneous polynomials of degree $k$  which are quasi-invariant with respect to $\A$.
It will be convenient to fix different notations for various parts of the series \eqref{poincare2} as follows.
Define
\begin{equation*} P_{\mathcal A}^{k,l}(t)=P^{k,l}=\sum_{i=k}^lb_it^i,
\end{equation*}
\begin{equation*} P_{{\mathcal A}, odd }^{k,l} (t)=P^{k,l}_{odd}=\sum_{k\leq 2i+1 \leq l}b_{2i+1}t^{2i+1},
\end{equation*}
and
\begin{equation*} P_{{\mathcal A}, even}^{k,l} (t) =P^{k,l}_{even}=\sum_{k \leq 2i \leq l}b_{2i}t^{2i}.
\end{equation*}

\begin{lemma}\label{pg1} Let $k \in \N$ satisfy $0 \leq k \leq n$. Then
$b_k=0$ if $k$ is odd, and $b_k=1$ if $k$ is even, in the Hilbert series \eqref{poincare2}.
\end{lemma}
\begin{proof} Let us first suppose that $2m \leq k$. Let $q$ be a homogeneous polynomial of degree $k$:
\begin{equation*}
q=\sum_{i=0}^{k}\lambda_{i}x^{k-i}y^{i},
\end{equation*}
for some $\lambda_i \in \mathbb{C}$. We know that $\lambda_1=\lambda_3=\ldots=\lambda_{2m-1}=0$ due to the quasi-invariance condition for the vector $\beta_0$. The quasi-invariance conditions for the vectors $\beta_j$ for $1 \leq j \leq n$ can be expressed in the matrix form $AC=0$,
where $C^T=(\lambda_0, \lambda_2,\ldots, \lambda_{2m-2}, \lambda_{2m}, \lambda_{2m+1},\ldots, \lambda_k)$
and the matrix $A$ consists of the columns $A_0, A_2, \ldots ,A_{2m-2}$, $A_{2m}, A_{2m+1}, \ldots ,A_{k}$ given by 
\begin{equation*}
A_{i}=\left( \begin{array}{c}
(k-i)\alpha_1^{k-i-1}-i\alpha_1^{k-i+1}\\
(k-i)\alpha_2^{k-i-1}-i\alpha_2^{k-i+1}\\
(k-i)\alpha_3^{k-i-1}-i\alpha_3^{k-i+1}\\
\vdots \\
(k-i)\alpha_n^{k-i-1}-i\alpha_n^{k-i+1}\\
\end{array}\right).
\end{equation*}

Suppose that $\sum_{i=0}^{m}a_{2i}A_{2i}+\sum_{i=2m+1}^ka_iA_i=0$ for some $a_i \in \mathbb{C}$. It follows that the polynomial
$
p(x)=\sum_{j=0}^{k-1} c_j x^j
$
satisfies $p(\alpha_j)=0$ for all $1 \leq j \leq n$, where we defined
\begin{equation}\label{small-deg-eqn1}
c_j=(j+1)a_{k-j-1}-(k-j+1) a_{k-j+1}
\end{equation}
and $a_{-1}=a_{1}=a_3=\ldots=a_{2m-1}=a_{k+1}=0$. Since $\deg p \leq k-1<n$ we conclude $p(x)\equiv 0$.
We have $c_j=0$ for $0 \leq j \leq k-1$. Suppose $k$ is odd. Then it follows from (\ref{small-deg-eqn1}) that $a_i=0$ for all $i=0,\ldots,k$.
Thus $A$ has rank $k+1-m$ and $b_k=0$. If $k$ is even then we have
\begin{equation*}a_1=a_3= a_5 = \ldots = a_{k-1}=0
\end{equation*}
and
$a_0\sim a_2 \sim a_4 \sim \ldots \sim a_k$,
where $\sim$ denotes proportionality with a non-zero coefficient. Thus $A$ has rank at least $k-m$. In fact the rank of $A$ is exactly $k-m$. This is because in this case one can produce a linear dependence in the  columns of $A$ by applying elementary column transformations to the columns $A_{2j}$, $0 \leq j \leq k/2$. Hence $b_k=1$. In the case $2m>k$ the arguments are similar.
\end{proof}

\begin{corollary}\label{pg1-cor}
The initial part of the Hilbert series is given by
\begin{equation*}
P^{0,n}=\frac{1-t^{2[\frac{n+2}{2}]}}{1-t^2}.
\end{equation*}
\end{corollary}


\begin{lemma}\label{lemma2m+n-1}
In the Hilbert series \eqref{poincare2} we have $b_{2m+n-1}=m$ if $n$ is even. If $n$ is odd then $b_{2m+n-2}=m-1$.
\end{lemma}
\begin{proof} Suppose $n$ is even, the odd case is similar. Let $q$ be a homogeneous polynomial of degree $2m+n-1$:
\begin{equation*}
q=\sum_{i=0}^{2m+n-1}a_{i}x^{2m+n-1-i}y^{i}
\end{equation*}
for some $a_i \in \mathbb{C}$. It follows from the quasi-invariance conditions for the vector $\beta_0=(0,1)$ that is from $\partial_y^s q|_{y=0}=0$  for $s=1,3,\ldots 2m-1$ that $a_i=0$ for $i=1,3, \ldots ,2m-1$.

The quasi-invariance conditions for the vectors $\beta_j$ for $1 \leq j \leq n$ can be expressed in the matrix form $AC=0$,
where
$C^T=(a_0, a_2, \ldots, a_{2m}, a_{2m+1}, a_{2m+2},\ldots, a_{2m+n-1})$
and the matrix $A$ consists of $m+n$ columns $A_0, A_2, \ldots ,A_{2m}, A_{2m+1}, \ldots, A_{2m+n-1}$ given by the following. For $0 \leq i \leq m$ one has
\begin{equation*}
A_{2i}=\left( \begin{array}{c}
(2m+n-1-2i)\alpha_1^{2m+n-2i-2}-2i\alpha_1^{2m+n-2i}\\
(2m+n-1-2i)\alpha_2^{2m+n-2i-2}-2i\alpha_2^{2m+n-2i}\\
(2m+n-1-2i)\alpha_3^{2m+n-2i-2}-2i\alpha_3^{2m+n-2i}\\
\vdots \\
(2m+n-1-2i)\alpha_n^{2m+n-2i-2}-2i\alpha_n^{2m+n-2i}\\
\end{array}\right),
\end{equation*}
while for $2m+1 \leq i \leq 2m+n-1$ one has
\begin{equation*}
A_{i}=(-1)^i\left( \begin{array}{c}
(2m+n-1-i)\alpha_1^{2m+n-i-2}-i\alpha_1^{2m+n-i}\\
(2m+n-1-i)\alpha_2^{2m+n-i-2}-i\alpha_2^{2m+n-i}\\
(2m+n-1-i)\alpha_3^{2m+n-i-2}-i\alpha_3^{2m+n-i}\\
\vdots \\
(2m+n-1-i)\alpha_n^{2m+n-i-2}-i\alpha_n^{2m+n-i}\\
\end{array}\right).
\end{equation*}
By applying appropriate elementary column transformations we reduce the matrix $A$ to the form
\begin{equation*}
{\small
\left( \begin{array}{cccccccccc}
\alpha_1^{2m+n-2} & \alpha_1^{2m+n-4} & \ldots & \alpha_1^2 & 1 & \alpha_1^{n-1} &\alpha_1^{n-3}\ldots &\alpha_1^3 &\alpha_1\\
\alpha_2^{2m+n-2} & \alpha_2^{2m+n-4} & \ldots & \alpha_2^2 & 1 & \alpha_2^{n-1} &\alpha_2^{n-3}\ldots &\alpha_2^3 &\alpha_2\\
\alpha_3^{2m+n-2} & \alpha_3^{2m+n-4} & \ldots & \alpha_3^2 & 1 & \alpha_3^{n-1} &\alpha_3^{n-3}\ldots &\alpha_3^3 &\alpha_3\\
\vdots&\vdots&&\vdots&\vdots&\vdots&\vdots&&\vdots \\
\alpha_n^{2m+n-2} & \alpha_n^{2m+n-4} & \ldots & \alpha_n^2 & 1 & \alpha_n^{n-1} &\alpha_n^{n-3}\ldots &\alpha_n^3 &\alpha_n\\
\end{array}\right).
}
\end{equation*}
Note that after the column permutation we have the $n \times n$ Vandermonde minor
\begin{equation*}\label{minorfor2m+n-1}
{\small
Q=\left| \begin{array}{cccccccc}
\alpha_1^{n-1}&\alpha_1^{n-2}&\alpha_1^{n-3}&\ldots&\alpha_1^3&\alpha_1^2&\alpha_1&1\\
\alpha_2^{n-1}&\alpha_2^{n-2}&\alpha_2^{n-3}&\ldots&\alpha_2^3&\alpha_2^2&\alpha_2&1\\
\alpha_3^{n-1}&\alpha_3^{n-2}&\alpha_3^{n-3}&\ldots&\alpha_3^3&\alpha_3^2&\alpha_3&1\\
\vdots&\vdots&\vdots&&\vdots&\vdots&\vdots&\vdots \\
\alpha_n^{n-1}&\alpha_n^{n-2}&\alpha_n^{n-3}&\ldots&\alpha_n^3&\alpha_n^2&\alpha_n&1\\
\end{array}\right|.
}
\end{equation*}
Since $Q \neq 0$ we get $\rank A=n$. Hence $b_{2m+n-1}=m+n-\rank A=m$ as stated.
\end{proof}
It is easy to find all the odd coefficients of the Hilbert series for terms larger than a certain degree.
\begin{proposition}\label{corronodd2m+n-1}  Let $i \geq 2m+n-1$ with $i$ odd. Then $b_{i}=i+1-m-n$ in the Hilbert series (\ref{poincare2}).
\end{proposition}
\begin{proof} As in the proof of Lemma \ref{lemma2m+n-1} we have $b_i=i+1-m-\rank A$, where $A$ is the matrix of the system of quasi-invariance conditions as equations for the coefficients of a polynomial of degree $i$. By the same reasons as in Lemma \ref{lemma2m+n-1} we get $\rank A=n$, so the statement follows.
\end{proof}

\begin{corollary}\label{pg2} For any configuration $\mathcal{A}$ of type $(m,1^n)$ we have
\begin{align*}&P_{odd }^{2m+n+1, 2m+2n-3}
&=\frac{(m+2)t^{2m+n+1} -mt^{2m+n+3} -(m+n)t^{2m+2n-1}+(m+n-2)t^{2m+2n+1}}{(t^2-1)^2}
\end{align*}
if $n$ is even, and
\begin{align*}&P_{odd }^{2m+n,2m+2n-3}
&=\frac{(m+1)t^{2m+n}-(m-1)t^{2m+n+2}-(m+n)t^{2m+2n-1}+ (m+n-2)t^{2m+2n+1} }{(t^2-1)^2}
\end{align*}
if $n$ is odd.
\end{corollary}

We need some lemmas before analyzing the even terms of the Hilbert series.
Let $\Delta= \prod_{1\leq i<j\leq [n/2]}(\alpha_i^2-\alpha_j^2)$. Refer to the elementary symmetric polynomials in the variables $\alpha_1^2, \alpha_2^2, \ldots ,\alpha_{[n/2]}^2$ as $\widehat{e}_i, 0 \leq i \leq [n/2]$, ($\widehat{e}_0=1$). Let $B$ be the $[n/2] \times (n-s)$ matrix, where $0 \leq s \leq n-1$, with the columns
\begin{equation}\label{CiX}C_i(X)=
\left( \begin{array}{c}
(2i-1)\alpha_1^{2i-2}-(2X+2n-2i+1)\alpha_1^{2i}\\
(2i-1)\alpha_2^{2i-2}-(2X+2n-2i+1)\alpha_2^{2i}\\
(2i-1)\alpha_3^{2i-2}-(2X+2n-2i+1)\alpha_3^{2i}\\
\vdots \\
(2i-1)\alpha_{[n/2]}^{2i-2}-(2X+2n-2i+1)\alpha_{[n/2]}^{2i}\\
\end{array}\right)
\end{equation}
where $1 \leq i \leq n-s$.
\begin{lemma}\label{construction1} Let $s \leq [n/2]$. Let $B_{L}$ be the minor of $B$ formed by taking the determinant of the square submatrix with columns $C_{L}(X), C_{L+1}(X), \ldots ,C_{L+[n/2]-1}(X)$ where $1 \leq L \leq [(n+1)/2]-s+1$. Then
\begin{align}\label{BL}
B_L=&\Delta \prod_{i=1}^{[n/2]}\alpha_i^{2L-2}\sum_{i=0}^{[n/2]}(-1)^i\prod_{r=i+1}^{[n/2]}(2[n/2]-2r+2L-1)\notag \\
&\times \prod_{r=1}^{i}(2X+2r+2\left[\frac{n+1}{2}\right]-2L+1) \widehat{e}_i.
\end{align}
\end{lemma}
\begin{proof} For $1 \leq i \leq n-s$ let us introduce the column vectors
\begin{equation*}x_{i}=(i+1)\left( \begin{array}{c}
\alpha_1^{i}\\
\alpha_2^{i}\\
\alpha_3^{i}\\
\vdots \\
\alpha_{[n/2]}^{i}
\end{array}\right),\,\, y_{i}=(2X+2n-i+1)\left( \begin{array}{c}
\alpha_1^{i}\\
\alpha_2^{i}\\
\alpha_3^{i}\\
\vdots \\
\alpha_{[n/2]}^{i}
\end{array}\right).
\end{equation*}
Thus $C_i(X)=x_{2i-2}-y_{2i}$.
Consider $B_1$ and suppose $n$ is even, odd case is similar.
Notice that
$$
B_1=\sum_{t=0}^{n/2}(-1)^tB_1^t,
$$
where
$B_1^t$ is  the determinant of the matrix with the columns $y_n, y_{n-2}, \ldots ,y_{n-2t+2}, x_{n-2t-2}, \ldots ,x_0$. We have
\begin{align*}B_1^t&=\prod_{r=t}^{n/2-1}(n-2r-1)\prod_{r=0}^{t-1}(2X+2r+n+1)\left |\begin{array}{cccccc}
\alpha_1^{n}&\ldots&\alpha_1^{n-2t+2}&\alpha_1^{n-2t-2}&\ldots&1\\
\alpha_2^{n}&\ldots&\alpha_2^{n-2t+2}&\alpha_2^{n-2t-2}&\ldots&1\\
\vdots&&\vdots&\vdots&&\vdots\\
\alpha_{n/2}^{n}&\ldots&\alpha_{n/2}^{n-2t+2}&\alpha_{n/2}^{n-2t-2}&\ldots&1\\
\end{array}\right|\\
&=\prod_{r=t}^{n/2-1}(n-2r-1)\prod_{r=0}^{t-1}(2X+2r+n+1)\Delta \widehat{e}_t
\end{align*}
since the elementary symmetric polynomials $\widehat{e}_t$ are particular Schur polynomials.
Hence $B_1$ has the required form. It is easy to see that we can adopt the same strategy used to expand $B_1$ to deal with each $B_{L}$, $1 \leq L \leq n/2-s+1$.
\end{proof}

\begin{lemma}\label{construction3} Let $s=0$ and let $B_L$ be defined by (\ref{BL}) with $X\in \mathbb{R}_{\ge 0}$. Suppose that $\Delta \prod_{i=1}^{[n/2]}\alpha_i \ne  0$. Then for even $n$ $\exists L, 1 \leq L \leq n/2+1$ s.t. $B_L \neq 0$. For odd $n$ $\exists L, 2 \leq L \leq \frac{n+3}{2}$ s.t. $B_L \neq 0$.
\end{lemma}
\begin{proof} By Lemma \ref{construction1} we have
{\small
\begin{align}\label{BL2}
B_L=&\Delta \prod_{i=1}^{[n/2]}\alpha_i^{2L-2}\sum_{i=0}^{[n/2]}(-1)^i\prod_{r=i+1}^{[n/2]}(2[n/2]-2r+2L-1)
 \prod_{r=1}^{i}(2X+2r+2\left[\frac{n+1}{2}\right] -2L+1)\widehat{e}_i.
\end{align}
}
Suppose $B_L=0$ for $1 \leq L \leq n/2+1$ if $n$ is even and suppose $B_L=0$ for $2 \leq L \leq \frac{n+3}{2}$ if $n$ is odd. Let us cancel $\Delta \prod_{i=1}^{[n/2]}\alpha_i^{2L-2}$ and consider the resulting conditions as a system of linear equations for the unknowns $\widehat{e}_0, \widehat{e}_1, \ldots ,\widehat{e}_{[n/2]}$. Refer to the corresponding matrix as $Q$. We will show that the determinant $|Q|\neq 0$, which would be a contradiction as  $\widehat e_0\ne 0$. We consider $|Q|$ as a polynomial in $X$. First we show $|Q|$ is not identically zero in $X$. Set $2X+2[n/2]+1=0$. Let us enumerate the rows of $Q$ by $L=1,2, \ldots ,n/2+1$ when $n$ is even and $L=2, \ldots ,\frac{n+3}{2}$ when $n$ is odd,  and the columns of $Q$ by $i=0,1,\ldots, [n/2]$. Then it follows from (\ref{BL2}) that the first $i$ entries in the $i$-th column are 0.  So $|Q|$ is the product of the diagonal entries  and this is non-zero.

Next we show that there are no positive values of $X$ for which $|Q|=0$. Note that as a polynomial in $X$, $|Q|$ has degree $\sum_{i=0}^{[n/2]}i=\frac{[n/2]([n/2]+1)}{2}$. Let us subtract the $(i+1)$st column from the $i$th column of $Q$, $i=0, 1, \ldots ,[n/2]-1$. Then the $L$th entry of the $i$th column is given by
\begin{align*}&(-1)^i\prod_{r=i}^{[n/2]-1}(2[n/2]-2r+2L-3)\prod_{r=0}^{i-1}(2X+2r+2[(n+1)/2]-2L+3)\\
&+(-1)^i\prod_{r=i+1}^{[n/2]-1}(2[n/2]-2r+2L-3)\prod_{r=0}^{i}(2X+2r+2[(n+1)/2]-2L+3)\\
&=(-1)^i\prod_{r=i+1}^{[n/2]-1}(2[n/2]-2r+2L-3)\prod_{r=0}^{i-1}(2X+2[(n+1)/2]-2L+3+2r)\\
&\times 2(X+n).
\end{align*}
We can repeat this process (subtracting the $(i+1)$st column from the $i$th column of $Q$ where $i=0, 1 \ldots ,[n/2]-k$ at the $k$-th iteration) to see that the expression
\begin{align*}\prod_{i=1}^{[n/2]}(X+n+1-i)^{[n/2]-i+1}
\end{align*}
is a factor of $|Q|$. This expression has the same total degree in $X$ as $|Q|$ and is non-zero for  $X\ge 0$. So $|Q|\neq 0$ and we are done.
\end{proof}


Now we are in the position to determine all the even coefficients of the Hilbert series starting with a certain degree.
\begin{proposition}\label{corron2m+2n} Let $i=2(m+n+t)$, where $t \in \mathbb{Z}_{\geq 0}$. Then  $b_{i}=i+1-m-n$ in the Hilbert series  (\ref{poincare2}).
\end{proposition}
\begin{proof}
Let $q$ be a homogeneous polynomial of degree $2m+2n+2t$,
\begin{equation*}
q=\sum_{i=0}^{2m+2n+2t}a_{i}x^{2m+2n+2t-i}y^{i},
\end{equation*}
where $a_i \in \mathbb{C}$ and $a_{2j-1}=0$ for $1 \le j \le m$. Consider the quasi-invariance conditions for $q$ for the vectors $\beta_1,\ldots,\beta_n$. The matrix of the corresponding system of linear equations for the coefficients of $q$ has the following structure after elementary transformations
{\small
\begin{equation}\label{2m+2n}
\left( \begin{array}{cccccc|c}
\alpha_1^{2m+2n+2t-1} & \alpha_1^{2m+2n+2t-3} & \ldots & \alpha_1^5 & \alpha_1^3 & \alpha_1&\multirow{5}{*}{$A$}\\
\alpha_2^{2m+2n+2t-1} & \alpha_2^{2m+2n+2t-3} & \ldots & \alpha_2^5 & \alpha_2^3 & \alpha_2&\\
\alpha_3^{2m+2n+2t-1} & \alpha_3^{2m+2n+2t-3} & \ldots & \alpha_3^5 & \alpha_3^3 & \alpha_3&\\
\vdots&\vdots&\ldots&\vdots&\vdots&\vdots \\
\alpha_n^{2m+2n+2t-1} & \alpha_n^{2m+2n+2t-3} & \ldots & \alpha_n^5 & \alpha_n^3 & \alpha_n&\\
\end{array}\right),
\end{equation}
}
where the block $A$ consists of $n+t$ columns $A_i$
\begin{equation*}
A_i=\left( \begin{array}{c}
(2i-1)\alpha_1^{2i-2}-(2m+2n+2t-2i+1)\alpha_1^{2i}\\
(2i-1)\alpha_2^{2i-2}-(2m+2n+2t-2i+1)\alpha_2^{2i}\\
(2i-1)\alpha_3^{2i-2}-(2m+2n+2t-2i+1)\alpha_3^{2i}\\
\vdots \\
(2i-1)\alpha_n^{2i-2}-(2m+2n+2t-2i+1)\alpha_n^{2i}\\
\end{array}\right)
\end{equation*}
with $1 \leq i \leq n+t$. Suppose the rank of the matrix (\ref{2m+2n}) is not full, so the rank is less than $n$.  We are going to show that at least one $n \times n$ minor is non-zero, a contradiction which implies the rank of the original matrix is in fact $n$. First we assume that
\begin{equation}\label{assumption}
\alpha_i=-\alpha_{[n/2]+i}, \, 1 \leq i \leq [n/2], \quad  \text{and} \,\,  \alpha_n=0 \,\,  \text{if} \,\,  n \,\, \text{is odd}.
\end{equation}
This assumption will be justified later. Then for odd $n$ the matrix (\ref{2m+2n}) can be rearranged to the matrix
\begin{equation}\label{2m+2nwithallpairs}
{\small
\left( \begin{array}{cccccc|c}
\alpha_1^{2m+2n+2t-1} & \alpha_1^{2m+2n+2t-3} & \ldots & \alpha_1^5 & \alpha_1^3 & \alpha_1&\multirow{5}{*}{0}\\
\alpha_2^{2m+2n+2t-1} & \alpha_2^{2m+2n+2t-3} & \ldots & \alpha_2^5 & \alpha_2^3 & \alpha_2&\\
\alpha_3^{2m+2n+2t-1} & \alpha_3^{2m+2n+2t-3} & \ldots & \alpha_3^5 & \alpha_3^3 & \alpha_3&\\
\vdots&\vdots&\ldots&\vdots&\vdots&\vdots \\
\alpha_{[n/2]}^{2m+2n+2t-1} & \alpha_{[n/2]}^{2m+2n+2t-3} & \ldots & \alpha_{[n/2]}^5 & \alpha_{[n/2]}^3 & \alpha_{[n/2]}&\\
\hline
&&\multirow{2}{*}{0}&&&&$ $\mathcal B$ $
\end{array}\right)
},
\end{equation}
where block $\mathcal B$ consists of $n+t$ columns of the form
\begin{equation}\label{blockB} {\mathcal C}_i=
\left( \begin{array}{c}
(2i-1)\alpha_1^{2i-2}-(2m+2n+2t-2i+1)\alpha_1^{2i}\\
(2i-1)\alpha_2^{2i-2}-(2m+2n+2t-2i+1)\alpha_2^{2i}\\
(2i-1)\alpha_3^{2i-2}-(2m+2n+2t-2i+1)\alpha_3^{2i}\\
\vdots \\
(2i-1)\alpha_{[n/2]}^{2i-2}-(2m+2n+2t-2i+1)\alpha_{[n/2]}^{2i}\\
(2i-1)0^{2i-2}\\
\end{array}\right)
\end{equation}
where $1 \leq i \leq n+t$. When $n$ is even the block $\mathcal B$ and the columns ${\mathcal C}_i$ have the same form (\ref{blockB}) with the final row removed. In this case the matrix $\mathcal B$ with columns ${\mathcal C}_i$ coincides with the matrix $B$ with columns $C_i(X)$ defined by  \eqref{CiX} with $X=m+t$.
By Lemma \ref{construction3}  the block $\mathcal B$ contains at least one non-zero $n/2 \times n/2$ minor when $n$ is even. When $n$ is odd $\mathcal B$ contains a non-zero $\frac{n+1}{2} \times \frac{n+1}{2}$ minor which contains the last row and the first column of $\mathcal B$.

In order to show that the rank of the original matrix (\ref{2m+2n}) is $n$ it is left to justify the assumption  \mref{assumption}. Since the block $\mathcal B$ contains at least one non-zero $[\frac{n+1}{2}]\times [\frac{n+1}{2}]$ minor, it contains at least one non-zero minor of each size $k \times k$, where $1 \leq k \leq [\frac{n+1}{2}]$. Suppose the rank of matrix
\mref{2m+2n}
is not $n$ so any $n \times n$ minor is zero. Consider
\begin{align*}
Q_1=&\left| \begin{array}{cccccc}
\alpha_1^{2n-1}&\alpha_1^{2n-3}&\alpha_1^{2n-5}&\ldots&\alpha_1^3&\alpha_1\\
\alpha_2^{2n-1}&\alpha_2^{2n-3}&\alpha_2^{2n-5}&\ldots&\alpha_2^3&\alpha_2\\
\alpha_3^{2n-1}&\alpha_3^{2n-3}&\alpha_3^{2n-5}&\ldots&\alpha_3^3&\alpha_3\\
\vdots\\
\alpha_{n}^{2n-1}&\alpha_{n}^{2n-3}&\alpha_n^{2n-5}&\ldots&\alpha_n^3&\alpha_n
\end{array}\right|\\
=\prod_{i=1}^n\alpha_i&\left| \begin{array}{cccccc}
\alpha_1^{2n-2}&\alpha_1^{2n-4}&\alpha_1^{2n-6}&\ldots&\alpha_1^2&1\\
\alpha_2^{2n-2}&\alpha_2^{2n-4}&\alpha_2^{2n-6}&\ldots&\alpha_2^2&1\\
\alpha_3^{2n-2}&\alpha_3^{2n-4}&\alpha_3^{2n-6}&\ldots&\alpha_3^2&1\\
\vdots\\
\alpha_{n}^{2n-2}&\alpha_{n}^{2n-4}&\alpha_n^{2n-6}&\ldots&\alpha_n^2&1
\end{array}\right|.
\end{align*}
So $Q_1=0$ implies that   $\alpha_i^2=\alpha_j^2$ for some $1 \leq i <j \leq n$, or $\alpha_k=0$ for some $k, 1 \leq k \leq n$. Suppose first that $\prod_{i<j} (\alpha_i^2-\alpha_j^2)\ne 0$. Then after relabelling of the indecies $\alpha_1=0$. Hence the original matrix has the following form
\begin{equation*}
{\small
\left( \begin{array}{cccccc|c}
0 & 0 & \ldots & 0 & 0 & 0&\multirow{5}{*}{A}\\
\alpha_2^{2m+2n+2t-1} & \alpha_2^{2m+2n+2t-3} & \ldots & \alpha_2^5 & \alpha_2^3 & \alpha_2&\\
\alpha_3^{2m+2n+2t-1} & \alpha_3^{2m+2n+2t-3} & \ldots & \alpha_3^5 & \alpha_3^3 & \alpha_3&\\
\vdots&\vdots&\ldots&\vdots&\vdots&\vdots \\
\alpha_n^{2m+2n+2t-1} & \alpha_n^{2m+2n+2t-3} & \ldots & \alpha_n^5 & \alpha_n^3 & \alpha_n&\\
\end{array}\right),
}
\end{equation*}
where the block $A$ consists of $n+t$ columns each with the following structure
\begin{equation*}
D_i=\left( \begin{array}{c}
(2i-1)0^{2i-2}\\
(2i-1)\alpha_2^{2i-2}-(2m+2n+2t-2i+1)\alpha_2^{2i}\\
(2i-1)\alpha_3^{2i-2}-(2m+2n+2t-2i+1)\alpha_3^{2i}\\
\vdots \\
(2i-1)\alpha_n^{2i-2}-(2m+2n+2t-2i+1)\alpha_n^{2i}\\
\end{array}\right)
\end{equation*}
with $1 \leq i \leq n+t$. In this situation the $n\times n$ determinant
\begin{align*}
Q_{1}^{'}=&\left| \begin{array}{ccccc|c}
0&0&0&\ldots&0&\multirow{5}{*}{$D_1$}\\
\alpha_2^{2n-3}&\alpha_2^{2n-5}&\alpha_2^{2n-7}&\ldots&\alpha_2&\\
\alpha_3^{2n-3}&\alpha_3^{2n-5}&\alpha_3^{2n-7}&\ldots&\alpha_3&\\
\vdots&\vdots&\vdots&&\vdots\\
\alpha_{n}^{2n-3}&\alpha_{n}^{2n-5}&\alpha_n^{2n-7}&\ldots&\alpha_n&\\
\end{array}\right| \neq 0.
\end{align*}
So we may assume that $\alpha_1^2=\alpha_2^2$ after relabelling. In this situation the matrix (\ref{2m+2n}) is equivalent by row transformations to the matrix
\begin{equation}\label{2m+2n+2tstage3}
{\small
\left( \begin{array}{cccccc|c}
0 & 0 & \ldots & 0 & 0 & 0&\multirow{5}{*}{$\widehat{A}$}\\
\alpha_1^{2m+2n+2t-1} & \alpha_1^{2m+2n+2t-3} & \ldots & \alpha_1^5 & \alpha_1^3 & \alpha_1&\\
\alpha_3^{2m+2n+2t-1} & \alpha_3^{2m+2n+2t-3} & \ldots & \alpha_3^5 & \alpha_3^3 & \alpha_3&\\
\vdots&\vdots&\ldots&\vdots&\vdots&\vdots \\
\alpha_n^{2m+2n+2t-1} & \alpha_n^{2m+2n+2t-3} & \ldots & \alpha_n^5 & \alpha_n^3 & \alpha_n&\\
\end{array}\right)
}
\end{equation}
where the block $\widehat{A}$ consists of $n+t$ columns each with the following structure
\begin{equation*}
\left( \begin{array}{c}
(2i-1)\alpha_1^{2i-2}-(2m+2n+2t-2i+1)\alpha_1^{2i}\\
0\\
(2i-1)\alpha_3^{2i-2}-(2m+2n-2i+2t+1)\alpha_3^{2i}\\
\vdots \\
(2i-1)\alpha_n^{2i-2}-(2m+2n-2i+2t+1)\alpha_n^{2i}\\
\end{array}\right)
\end{equation*}
with $1 \leq i \leq n+t$. Note that the rows $1, 3 \ldots ,[n/2], [n/2]+1$ of $\widehat{A}$ coincide up to relabelling of the $\alpha_i$ with the matrix $\widehat{B}$ formed by the first $[n/2]$ rows of the matrix $\mathcal B$ from (\ref{2m+2nwithallpairs}), (\ref{blockB}). As $\rank {\mathcal B}=[\frac{n+1}{2}]$ and $\rank \widehat B = [n/2]$ the first row of the matrix $\widehat{A}$ contains some non-zero entry which we denote by $B_{1 \times 1}$. Now, consider the following minor of the matrix (\ref{2m+2n+2tstage3})
\begin{align*}
Q_2=&\left| \begin{array}{ccccccc}
0&0&0&\ldots&0&0&B_{1 \times 1}\\
\alpha_1^{2n-3}&\alpha_1^{2n-5}&\alpha_1^{2n-7}&\ldots&\alpha_1^3&\alpha_1&0\\
\alpha_3^{2n-3}&\alpha_3^{2n-5}&\alpha_3^{2n-7}&\ldots&\alpha_3^3&\alpha_3&\multirow{4}{*}{*}\\
\alpha_4^{2n-3}&\alpha_4^{2n-5}&\alpha_4^{2n-7}&\ldots&\alpha_4^3&\alpha_4&\\
\vdots&\vdots&\vdots&&\vdots&\vdots&\\
\alpha_{n}^{2n-3}&\alpha_{n}^{2n-5}&\alpha_n^{2n-7}&\ldots&\alpha_n^3&\alpha_n&
\end{array}\right|.
\end{align*}
Since $B_{1 \times 1}\neq 0$, $Q_2=0$ implies that $\alpha_i^2=\alpha_j^2$ for some $3 \leq i <j \leq n$ or $\alpha_k=0$ for $k \neq 1,2$. Suppose that $\prod_{i<j; i,j=3}^n (\alpha_i^2-\alpha_j^2)\ne 0$ so $\alpha_3=0$ up to relabelling. Then (\ref{2m+2n+2tstage3}) is equivalent to a matrix which has minor
\begin{align*}
Q_2^{'}=&\left| \begin{array}{ccccccc}
0&0&0&\ldots&0&0&\multirow{2}{*}{$B_{2 \times 2}$}\\
0&0&0&\ldots&0&0&\\
\alpha_1^{2n-5}&\alpha_1^{2n-7}&\alpha_1^{2n-9}&\ldots&\alpha_1^3&\alpha_1&\multirow{4}{*}{*}\\
\alpha_4^{2n-5}&\alpha_4^{2n-7}&\alpha_4^{2n-9}&\ldots&\alpha_4^3&\alpha_4&\\
\vdots&\vdots&\vdots&&\vdots&\vdots&\\
\alpha_{n}^{2n-5}&\alpha_{n}^{2n-7}&\alpha_n^{2n-9}&\ldots&\alpha_n^3&\alpha_n&
\end{array}\right|,\\
\end{align*}
where $B_{2 \times 2}$ is a $ 2 \times 2$ submatrix of the first two rows of A such that the corresponding minor $|B_{2 \times 2}|\neq 0$. Such a submatrix exists since $\rank \widehat{B}=[\frac{n}{2}]$.  So we have $Q_2^{'}\neq 0$. So we may assume that $\alpha_1^2=\alpha_2^2$ and $\alpha_3^2=\alpha_4^2$ up to relabelling. It is not hard to see that we can continue in this way to deduce that  (up to relabelling) $\alpha_i^2=\alpha_{[n/2]+i}^2, \, 1 \leq i \leq [n/2]$. If $n$ is even this justifies the assumption \mref{assumption}.
Suppose now that  $n$ is odd. Then the matrix (\ref{2m+2n}) is equivalent to
\begin{equation}\label{2m+2n+2tstage4}
{\small
\left( \begin{array}{cccccc|c}
\alpha_1^{2m+2n+2t-1} & \alpha_1^{2m+2n+2t-3} & \ldots & \alpha_1^5 & \alpha_1^3 & \alpha_1&\multirow{5}{*}{0}\\
\alpha_2^{2m+2n+2t-1} & \alpha_2^{2m+2n+2t-3} & \ldots & \alpha_2^5 & \alpha_2^3 & \alpha_2&\\
\alpha_3^{2m+2n+2t-1} & \alpha_3^{2m+2n+2t-3} & \ldots & \alpha_3^5 & \alpha_3^3 & \alpha_3&\\
\vdots&\vdots&\ldots&\vdots&\vdots&\vdots \\
\alpha_{\frac{n-1}{2}}^{2m+2n+2t-1} & \alpha_{\frac{n-1}{2}}^{2m+2n+2t-3} & \ldots & \alpha_{\frac{n-1}{2}}^5 & \alpha_{\frac{n-1}{2}}^3 & \alpha_{\frac{n-1}{2}}&\\
\hline
&&0&&&&\widehat{B}\\
\hline
\alpha_{n}^{2m+2n+2t-1} & \alpha_{n}^{2m+2n+2t-3} & \ldots & \alpha_{n}^5 & \alpha_{n}^3 & \alpha_{n}&*
\end{array}\right)
},
\end{equation}
where the block $\widehat{B}$ is up to relabelling the first $(n-1)/2$ rows of the block $\mathcal B$ introduced in (\ref{2m+2nwithallpairs}), (\ref{blockB}). Thus we can consider the following minor of (\ref{2m+2n+2tstage4})
\begin{equation}
{\small
\left( \begin{array}{ccccc|c}
\alpha_1^{n} & \alpha_1^{n-2} & \ldots  & \alpha_1^3 & \alpha_1&\multirow{5}{*}{0}\\
\alpha_2^{n} & \alpha_2^{n-2} & \ldots &  \alpha_2^3 & \alpha_2&\\
\vdots&\vdots&\ldots&\vdots&\vdots \\
\alpha_{\frac{n-1}{2}}^{n} & \alpha_{\frac{n-1}{2}}^{n-2} & \ldots  & \alpha_{\frac{n-1}{2}}^3 & \alpha_{\frac{n-1}{2}}&\\
\hline
&&0&&&B_{\frac{n-1}{2}\times\frac{n-1}{2}}\\
\hline
\alpha_{n}^{n} & \alpha_{n}^{n-2} & \ldots  & \alpha_{n}^3 & \alpha_{n}&*
\end{array}\right)
},
\end{equation}
where $B_{\frac{n-1}{2} \times \frac{n-1}{2}}$ corresponds to a non-zero $(n-1)/2 \times (n-1)/2$ minor of $\widehat{B}$. We have $Q_{n-1}^{'}=0$ implies that $\alpha_n=0$.
This justifies the assumption \mref{assumption} for odd $n$.
\end{proof}

\begin{corollary}\label{pg3} For any configuration $\mathcal{A}$ of type $(m,1^n)$ we have
$$
P^{2m+2n-1, \infty}=\frac{t^{2m+2n-1}}{(1-t)^2}(m+n-(m+n-1)t).
$$
\end{corollary}
\begin{proof} Using Proposition \ref{corronodd2m+n-1} and Proposition \ref{corron2m+2n} we have
{\small
\begin{align*} P^{2m+2n-1, \infty}
=\sum_{i \geq 2m+2n-1}(i+1-m-n)t^i=\sum_{i \geq 2m+2n-1}it^i-(m+n-1)\sum_{i \geq 2m+2n-1}t^i\\
=\frac{t^{2m+2n-1}(2m+2n-1-2(m+n-1)t)}{(1-t)^2} -
\frac{(m+n-1)t^{2m+2n-1}}{1-t}.
\end{align*}
}
\end{proof}

\section{Hilbert series of the algebra $Q_{\arrange}$} \label{poincare-am1n}

In this Section we complete the derivation of the Hilbert series of the algebra of quasi-invariants $Q_\A$ for the configuration $\A=\amn$. This configuration was defined in Section~\ref{mm1n}, Definition \ref{m1n} in terms of the elementary symmetric polynomials $e_i$ of the variables $z_j=e^{2i \varphi_j}$ where  the vectors of $\A$ had the form $(\cos \vf_j, \sin \vf_j)$. In Section \ref{quasi-preliminary} we studied the quasi-invariance conditions for the configurations of type $(m,1^n)$ consisting of vectors $\beta_i=(1,\alpha_i)$. We worked with such conditions by making use of the elementary symmetric polynomials $\widehat e_r$ of $\alpha_i^2$ (see Proposition \ref{corron2m+2n} and its proof, Lemma \ref{construction1} and the notation before the lemma). Thus we start with the rearranging the definition of the configuration $\amn$ in terms of quantities $\widehat e_r$. We keep notation of  the previous Sections. 

\begin{proposition}\label{eqofconfigs} Suppose a configuration $\mathcal{A}$ satisfies the symmetry properties \mref{assumption}
 and
\begin{equation} \label{er-hat-explicitely}
\widehat{e}_r=\binom{[n/2]}{r}\displaystyle\prod_{i=1}^r\frac{2[\frac{n+1}{2}]-2i+1}{2m+2i-1},
\end{equation}
where $0 \leq r \leq [n/2]$. Then $\mathcal{A}$ is equivalent to $\arrange$.
\end{proposition}

More exactly we show that the configuration $\mathcal A$ defined by Proposition \ref{eqofconfigs} coincides with the configuration $\amn$ rotated by $\pi/2$ after renumbering of vectors. That is  the parameters are related by $\alpha_j = \cot \vf_j$ for $j=1,\ldots, n$, and the vectors of $\amn$  of multiplicity 1 are renumbered so that $\vf_j+\vf_{[n/2]+j}\in \pi\Z$, for $1\le j \le [n/2]$.

We prove Proposition \ref{eqofconfigs} by establishing the following two lemmas. Define $u_j=\sin^2 \vf_j$ for $1 \leq j \leq [n/2]$. Denote by $f_i$, $0 \leq i \leq [n/2]$, the $i$-th elementary symmetric polynomial in the variables $u_j$.
 First we find in Lemma \ref{valuesforfi} the values of the elementary symmetric polynomials $f_i$ for the configuration $\arrange$. Then we check in Lemma \ref{valuesforeiarrange} that these $f_i$ lead to the relations  \mref{er-hat-explicitely}.

\begin{lemma}\label{valuesforfi} For the configuration $\amn$ the elementary symmetric polynomials $f_i$, $0 \leq i \leq [n/2]$, take the  values
\begin{equation}\label{fiunderehatreln1}f_i=\binom{[n/2]}{i}2^{-i}\prod_{s=1}^i\frac{2m+2[n/2]-2s+1}{(m+n-s)}.
\end{equation}
\end{lemma}
The outline of the proof of Lemma \ref{valuesforfi} is as follows. When $n$ is even one can check that
\begin{equation}\label{fiunderehatreln2} e_r=\sum_{i=0}^r(-1)^i2^{2i}\binom{n-2i}{r-i}f_i,
\end{equation}
where $0\le r \le n/2$.
This expression is derived via directly expressing the symmetric polynomials $e_r$ through the variables $u_i$. Then we have to show that
\begin{align}\label{identityAarrangements} (-1)^r\binom{n}{r}\binom{m+r-1}{r}\binom{m+n-1}{r}^{-1} =\sum_{i=0}^r(-1)^{i}2^{i}\binom{n-2i}{r-i}\binom{n/2}{i}\prod_{s=1}^i\frac{2m+n-2s+1}{m+n-s}.
\end{align}
Recall the Saalsch{\"{u}}tz's theorem (see e.g \cite{Gessel}) for the generalised hypergeometric function
\begin{equation}\label{Saals}
\,_3F_2(a,b,-p;c,1+a+b-c-p;1)=\frac{(c-a)_p(c-b)_p}{(c)_p(c-a-b)_p},
\end{equation}
where $(\alpha)_p=\Gamma(\alpha+p)/\Gamma(\alpha)$.
The identity (\ref{identityAarrangements}) follows upon specialising \mref{Saals} for $a=-r, b=r-n, p=m+\frac{n-1}{2}, c = - \frac{n-1}{2}$.
We refer for further details including the case of odd $n$ to \cite{DJ}.

\begin{lemma}\label{valuesforeiarrange}Suppose that for $0 \leq r \leq [n/2]$
\begin{equation*}f_r=\binom{[n/2]}{r}2^{-r}\prod_{s=1}^r\frac{2m+2[n/2]-2s+1}{(m+n-s)}.
\end{equation*}
Then for $0 \leq r \leq [n/2]$ we have
\begin{equation}
\widehat{e}_r=\binom{[n/2]}{r}\prod_{s=1}^r\frac{2[\frac{n+1}{2}]-2s+1}{2m+2s-1}.
\end{equation}
\end{lemma}
{\it Scheme of proof.}
The symmetric polynomials $\widehat{e}_r$ are elementary symmetric polynomials in $u_i^{-1}-1$, $1 \le i \le [n/2]$. They can be expressed through the elementary symmetric polynomials $f_i$ in the variables $u_i$ factored by $f_{[n/2]}=\prod_{i=1}^{[n/2]} u_i$. This leads to the formula
\begin{equation*}
\widehat{e}_r=\sum_{i=0}^r(-1)^i2^i\binom{[n/2]-i}{r-i}\binom{[n/2]}{i}\prod_{s=0}^{i-1}\frac{m+[\frac{n+1}{2}]+s}{2m+2s+1}.
\end{equation*}
So one has to check the identity
\begin{align}\label{identityBarrangements}
\sum_{i=0}^r(-1)^{r-i}2^i\binom{[n/2]-i}{r-i}\binom{[n/2]}{i}\prod_{s=0}^{i-1}\frac{m+[\frac{n+1}{2}]+s}{2m+2s+1} 
=\binom{[n/2]}{r}\prod_{s=1}^r\frac{2[\frac{n+1}{2}]-2s+1}{2m+2s-1}.
\end{align}
Notice that $\binom{[n/2]-i}{r-i}\binom{[n/2]}{i}=\binom{[n/2]}{r}\binom{r}{i} $ and cancel $\binom{[n/2]}{r}$ in both sides of \mref{identityBarrangements}.  After the cancellation each side of (\ref{identityBarrangements}) is a polynomial in $n$ of degree $r$ with the highest coefficient 1 both for even and odd $n$ cases. It is easy to see in each case that the roots of the corresponding pair of polynomials coincide which implies the Lemma.

We continue to study the quasi-invariance conditions for the configuration $\amn$ with the help of Proposition \ref{eqofconfigs}. Firstly we note the following technical result.

\begin{lemma}\label{construction2} Let
\begin{equation}\label{erintheconfigequality}
\widehat{e}_r=\binom{[n/2]}{r}\prod_{i=1}^r\frac{2[\frac{n+1}{2}]-2i+1}{2m+2i-1},
\end{equation}
where $0 \leq r \leq [n/2]$. For $1 \leq s \leq [n/2]$ and $1 \leq L \leq [(n+1)/2]-s+1$ let $B_L$ be defined by (\ref{BL}) with $X=m-s$. Then $B_L=0$.
\end{lemma}
\begin{proof} We have
{\small
\begin{align*}B_L&=\Delta \prod_{i=1}^{[n/2]}\alpha_i^{2L-2}\sum_{i=0}^{[n/2]}(-1)^i\prod_{r=i+1}^{[n/2]}(2[n/2]-2r+2L-1)\prod_{r=1}^{i}(2m-2s+2r+2[(n+1)/2]-2L+1)\widehat{e}_i\\
&=y \Delta \prod_{i=1}^{[n/2]}\alpha_i^{2L-2} \sum_{i=0}^{[n/2]}(-1)^i\binom{[n/2]}{i}U(i),
\end{align*}
}
where
\begin{equation*}y=\prod_{r=1}^{[n/2]}(2m+2r-1)^{-1} \prod_{r=[(n+1)/2]-s-L+1}^{[n/2]-1}(2m+2r+1)\prod_{r=2L-[(n+1)/2]}^{L}(2[(n+1)/2]-2L+2r-1),
\end{equation*}
and
\begin{equation*}U(x)=\prod_{r=0}^{[(n+1)/2]-s-L}(2m+2x+2r+1)\prod_{r=0}^{\tilde L}(2[(n+1)/2]-2x+2r+1),
\end{equation*}
with $\tilde L= L-2$ when $n$ is even, and $\tilde L= L-3$ when $n$ is odd.
Note that $U(x)$ is a polynomial of degree $[n/2]-s<[n/2]$. It follows that $B_L=0$ by a standard result on sums of binomial coefficients.
\end{proof}

We will also need the following few lemmas.

\begin{lemma}\label{corrtoconstruction3} Let $B_L$ be defined by (\ref{BL}) with $X\ge 0$, and suppose that $\Delta \prod_{i=1}^{[n/2]}\alpha_i\ne 0$. The system of equations $B_L= 0$ where $1 \leq L \leq n/2$ if $n$ is even and $2 \leq L \leq \frac{n+1}{2}$ if $n$ is odd as the system of linear equations for the unknowns $\widehat{e}_1, \ldots ,\widehat{e}_{[n/2]}$ has a unique solution.
\end{lemma}

The Lemma follows from the proof of Lemma \ref{construction3}.
Recall now the notation of Lemma~\ref{construction1}.

\begin{lemma}\label{construction7} Let $0 \leq s \leq n-1$. Fix $q\in \Z_{\ge 0}$ such that $q\ge s - [(n+1)/2]$.  Let $\kappa=\{k_1, k_2, \ldots ,k_{q}\}$ where $1 \leq k_i \leq [n/2]$ are such that $k_i \ne k_j$ if $i\ne j$.
Take $L \in \Z$ such that $1 \le L \le q+[(n+1)/2]-s+1$.
Let $D_L^{\kappa}$ be the minor of $B$ formed by taking  the square submatrix with columns $C_L(X), C_{L+1}(X), \ldots ,C_{L+[n/2]-q-1}(X)$ and we include all rows of $B$ except rows $k_1, k_2, \ldots k_{q}$. Then
\begin{align}\label{Ek}
D_L^{\kappa}=&\Delta_{\kappa} \prod_{\substack{i=1\\i\not\in \kappa}}^{[n/2]}\alpha_i^{2L-2}
\sum_{i=0}^{[n/2]-q}(-1)^i\prod_{r=i+1}^{[n/2]-q}(2[n/2]-2r-2q+2L-1)
\notag \\
&\times \prod_{r=1}^{i}(2X+2r+2q+2[(n+1)/2]-2L+1)\widehat{e}_i^{\kappa},
\end{align}
where $\Delta_\kappa= \prod_{\nad{i<j}{i,j \notin \kappa}}^{[n/2]} (\alpha_i^2-\alpha_j^2)$, and $\widehat{e}_i^{\kappa}$ denotes the $i$-th elementary symmetric polynomials in the variables $\alpha_i^2, 1 \leq i \leq [n/2]$, $i \ne k_j$ where $j=1,\ldots, q$. 
\end{lemma}
The proof is same as  the one of Lemma \ref{construction1}.

\begin{lemma}\label{construction8} In the assumptions of Lemma \ref{construction7} let $L$ be fixed. Suppose that  $q\ge 1$ and that
$\Delta \prod_{i=1}^{[n/2]} \alpha_i \ne 0$.
 Then $\exists \kappa$ such that $D_L^{\kappa} \neq 0$.
\end{lemma}
\begin{proof}
Suppose that $D_L^{\kappa}=0$ for the following collections of $\kappa=\kappa_j=\{1,2,\ldots, q-1, q+j-1\}$, where $j=1,  \ldots ,[n/2]-q+1$. We cancel the term $\Delta_\kappa \prod_{\substack{i=1\\i\not \in \kappa_j}}^{[n/2]}\alpha_i^{2L-2}$ in the equation $D_L^{\kappa_j}=0$, and consider the resulting conditions as a system of linear equations for the unknowns
\begin{align*}x_i=(-1)^i\prod_{r=i+1}^{[n/2]-q}(2[n/2]-2r-2q+2L-1)\prod_{r=1}^{i}(2X+2r+2q+2[(n+1)/2]-2L+1),
\end{align*}
where $i=0, \ldots ,[n/2]-q$. The system takes the matrix form
$AY=0$,
where $A=(a_{jl})$, $1 \leq j \leq [n/2]-q+1$, $0 \leq l \leq [n/2]-q$, $a_{jl}=\widehat{e}_l^{\kappa_j}$ and $Y=(x_{0}, \ldots ,x_{[n/2]-q})$. Note that the determinant of $A$ has degree
$$1+2+ \ldots +([n/2]-q)=\binom{[n/2]-q+1}{2}$$
as a polynomial in  $\alpha_q^2, \ldots, \alpha_{[n/2]}^2$. We claim that $\det A\neq 0$. This can be seen by setting $\alpha_l^2=\alpha_m^2$ for $q\leq l<m \leq [n/2]$. In this situation we have $\widehat{e}_i^{{\kappa}_{l-q+1}}=\widehat{e}_i^{{\kappa}_{m-q+1}}$ for all $ 0 \leq i \leq [n/2]-q$ and thus $\det A=0$. Since $\det A$ has degree $\binom{[n/2]-q+1}{2}$ it has no zeroes under our assumptions.  This is a contradiction as $x_0\ne 0$ so $Y\ne 0$. \end{proof}

\begin{lemma}\label{construction11} In the notation of Lemma \ref{construction7}, set $s=1$. Fix $q=1$ and $k \in \N$ such that $1\le k \le [n/2]$. Suppose $X\ge 0$ and $\Delta\prod_{i=1}^{[n/2]} \alpha_i\ne 0$.     Then $\exists L, 2 \leq L \leq [(n+1)/2]+1$ such that $D_L^k \neq 0$.
\end{lemma}
 The proof is same as the one of Lemma \ref{construction3}.


Now we investigate the dimension of quasi-invariants of degree $2(m+n-1)$. It appears that this dimension is the same for any configuration $\mathcal{A}$ of type $(m, 1^n)$ except for one configuration whose geometry is fully fixed, namely for the configuration $\arrange$.
\begin{proposition}\label{lemma2m+2n-2} Suppose $\mathcal{A}$ has type $(m,1^n)$. Then  $b_{2(m+n-1)}=m+n-1$ in the Hilbert series (\ref{poincare2}), unless $\mathcal{A}$ is equivalent to $\arrange$.
\end{proposition}
\begin{proof}
Let $q$ be a homogeneous quasi-invariant polynomial of degree $2(m+n-1)$, let
\begin{equation*}
q=\sum_{i=0}^{2m+2n-2}a_{i}x^{2m+2n-2-i}y^{i},
\end{equation*}
where $a_i \in \mathbb{C}$, and $a_{2j-1}=0$ for $1\le j \le m$. The matrix of the system of quasi-invariance conditionce for the vectors $\beta_1,\ldots,\beta_n$ for the coefficients of $q$ has the structure
\begin{equation}\label{2m+2n-2}
{\small
M=\left( \begin{array}{cccccc|c}
\alpha_1^{2m+2n-3} & \alpha_1^{2m+2n-5} & \ldots & \alpha_1^5 & \alpha_1^3 & \alpha_1&\multirow{5}{*}{A}\\
\alpha_2^{2m+2n-3} & \alpha_2^{2m+2n-5} & \ldots & \alpha_2^5 & \alpha_2^3 & \alpha_2&\\
\alpha_3^{2m+2n-3} & \alpha_3^{2m+2n-5} & \ldots & \alpha_3^5 & \alpha_3^3 & \alpha_3&\\
\vdots&\vdots&\ldots&\vdots&\vdots&\vdots \\
\alpha_n^{2m+2n-3} & \alpha_n^{2m+2n-5} & \ldots & \alpha_n^5 & \alpha_n^3 & \alpha_n&\\
\end{array}\right),
}
\end{equation}
where the block A consists of $n-1$ columns each with the following structure
\begin{equation*}
A_i=\left( \begin{array}{c}
(2i-1)\alpha_1^{2i-2}-(2m+2n-2i-1)\alpha_1^{2i}\\
(2i-1)\alpha_2^{2i-2}-(2m+2n-2i-1)\alpha_2^{2i}\\
(2i-1)\alpha_3^{2i-2}-(2m+2n-2i-1)\alpha_3^{2i}\\
\vdots \\
(2i-1)\alpha_n^{2i-2}-(2m+2n-2i-1)\alpha_n^{2i}\\
\end{array}\right)
\end{equation*}
with $1 \leq i \leq n-1$.  Assume initially that the symmetry assumption \mref{assumption} holds. Then the matrix $M$ is equivalent to
\begin{equation}\label{2m+2n-2withallpairs}
{\small
\left( \begin{array}{cccccc|c}
\alpha_1^{2m+2n-3} & \alpha_1^{2m+2n-5} & \ldots & \alpha_1^5 & \alpha_1^3 & \alpha_1&\multirow{5}{*}{0}\\
\alpha_2^{2m+2n-3} & \alpha_2^{2m+2n-5} & \ldots & \alpha_2^5 & \alpha_2^3 & \alpha_2&\\
\alpha_3^{2m+2n-3} & \alpha_3^{2m+2n-5} & \ldots & \alpha_3^5 & \alpha_3^3 & \alpha_3&\\
\vdots&\vdots&\ldots&\vdots&\vdots&\vdots \\
\alpha_{[n/2]}^{2m+2n-3} & \alpha_{[n/2]}^{2m+2n-5} & \ldots & \alpha_{[n/2]}^5 & \alpha_{[n/2]}^3 & \alpha_{[n/2]}&\\
\hline
&&\multirow{2}{*}{0}&&&& \mathcal B
\end{array}\right),
}
\end{equation}
where
the block $\mathcal B$ consists of $n-1$ columns, each of which has the following structure for odd $n$:
\begin{equation}\label{blockB2m+2n-2} {\mathcal C}_i=
\left( \begin{array}{c}
(2i-1)\alpha_1^{2i-2}-(2m+2n-2i-1)\alpha_1^{2i}\\
(2i-1)\alpha_2^{2i-2}-(2m+2n-2i-1)\alpha_2^{2i}\\
(2i-1)\alpha_3^{2i-2}-(2m+2n-2i-1)\alpha_3^{2i}\\
\vdots \\
(2i-1)\alpha_{[n/2]}^{2i-2}-(2m+2n-2i-1)\alpha_{[n/2]}^{2i}\\
(2i-1)0^{2i-2}
\end{array}\right)
\end{equation}
where $1 \leq i \leq n-1$. In the case of even $n$ the last row in ${\mathcal C}_i$  should be removed. By Lemma~\ref{construction11} there exists a non-zero $([n/2]-1) \times ([n/2]-1)$ minor $B_{([n/2]-1)\times([n/2]-1)}$ of ${B}$ where ${B}$ is the submatrix of $\mathcal B$ formed by the first $[n/2]$ rows. Thus it follows that $\rank M\ge n-1$. The existence of this minor $B_{([n/2]-1) \times ([n/2]-1)}$ also allows us to reproduce the arguments from the proof of Proposition \ref{corron2m+2n} which justify the assumption \mref{assumption}  if  $\rank M<n$. Thus it follows that $\rank M\ge n-1$, and that the symmetry \mref{assumption} holds if $\rank M=n-1$.

Suppose now that  $\rank M=n-1$ that is $b_{2(m+n-1)}\ne m+n-1$. Then any $n \times n$ minor of $M$ must vanish. By Lemma \ref{construction1} in its notation we have
\begin{align*}
B_L=&\Delta \prod_{i=1}^{[n/2]}\alpha_i^{2L-2}\sum_{i=0}^{[n/2]}(-1)^i\prod_{r=i}^{[n/2]-1}(2[n/2]-2r+2L-3)\\
&\times \prod_{r=0}^{i-1}(2m+2r+2n-2[n/2]-2L+1)\widehat{e}_i,
\end{align*}
where $1 \leq L \leq n/2$ if $n$ is even and $2 \leq L \leq \frac{n+1}{2}$ if $n$ is odd.
Note that for even $n$ $B_L$ equals the minor of $\mathcal B$ where the columns $L,\ldots,L+n/2-1$ are kept. In case of odd $n$ $B_L$ has same absolute value as the minor of $\mathcal B$ where the columns $1, L, \ldots, L+(n-1)/2-1$ are kept.

Thus $B_L=0$ for the specified range of $L$.  Regard the resulting conditions as a system of linear equations for the unknowns $\widehat{e}_0, \widehat{e}_1, \ldots , \widehat{e}_{[n/2]}$. Then by Lemma \ref{corrtoconstruction3} the solution is unique, and it is given by
\begin{equation*}\widehat{e}_r=\binom{[n/2]}{r}\displaystyle\prod_{i=1}^r\frac{2[(n+1)/2]-2i+1}{2m+2i-1}
\end{equation*}
for $0 \leq r \leq [n/2]$ by Lemma \ref{construction2}. Thus $\mathcal A$ is equivalent to $\amn$ by Proposition \ref{eqofconfigs}.
\end{proof}

Now we determine the even coefficients $b_i$ of the Hilbert series $P_{\amn}$ where $\min(2m+2, n+1) \leq i \leq 2m+2n-4$. These even coefficients remained unknown after the general analysis of Section \ref{quasi-preliminary}.

\begin{lemma}\label{arrangecase1} Let $\mathcal{A}=\arrange$. Let $i=2(m+n-s)$, where $1\le s \le [n/2]$. Then $b_i=i-m-n+2$ in the Hilbert series (\ref{poincare2}).
\end{lemma}
\begin{proof} Let $q$ be a homogeneous quasi-invariant of degree $i$. Recall that $\alpha_i^2=\alpha_{[n/2]+i}^2, \, 1 \leq i \leq [n/2]$, and $\alpha_n=0$ if $n$ is odd. Then the matrix $M$ expressing the quasi-invariant conditions as linear equations for the non-zero coefficients of $q$ is equivalent to
\begin{equation}\label{2m+2n-2rmatrix}
{\small
M=\left( \begin{array}{cccccc|c}
\alpha_1^{2m+2n-2s-1} & \alpha_1^{2m+2n-2s-3} & \ldots & \alpha_1^5 & \alpha_1^3 & \alpha_1&\multirow{5}{*}{0}\\
\alpha_2^{2m+2n-2s-1} & \alpha_2^{2m+2n-2s-3} & \ldots & \alpha_2^5 & \alpha_2^3 & \alpha_2&\\
\alpha_3^{2m+2n-2s-1} & \alpha_3^{2m+2n-2s-3} & \ldots & \alpha_3^5 & \alpha_3^3 & \alpha_3&\\
\vdots&\vdots&\ldots&\vdots&\vdots&\vdots \\
\alpha_{[n/2]}^{2m+2n-2s-1} & \alpha_{[n/2]}^{2m+2n-2s-3} & \ldots & \alpha_{[n/2]}^5 & \alpha_{[n/2]}^3 & \alpha_{[n/2]}&\\
\hline
&&\multirow{2}{*}{0}&&&&\mathcal B
\end{array}\right).
}
\end{equation}
The block $\mathcal B$ consists of $n-s$ columns ${\mathcal C}_1, \ldots , {\mathcal C}_{n-s}$ with the following structure
\begin{equation*}{\mathcal C}_j=
\left( \begin{array}{c}
(2j-1)\alpha_1^{2j-2}-(2m+2n-2j+1-2s)\alpha_1^{2j}\\
(2j-1)\alpha_2^{2j-2}-(2m+2n-2j+1-2s)\alpha_2^{2j}\\
(2j-1)\alpha_3^{2j-2}-(2m+2n-2j+1-2s)\alpha_3^{2j}\\
\vdots \\
(2j-1)\alpha_{[n/2]}^{2j-2}-(2m+2n-2j+1-2s)\alpha_{[n/2]}^{2j}\\
(2j-1)0^{2j-2}
\end{array}\right)
\end{equation*}
where $1 \leq j \leq n-s$, and the last row should be removed if $n$ is even. We are going to show at first that $\rank M<n$. It is clear that any $n\times n$ minor equals zero unless we take exactly $[\frac{n+1}{2}]$ columns from the block $\mathcal B$. Now, consider the block $\mathcal B$. Let $B_L$ be the minor formed by taking the determinant of the square submatrix with columns ${\mathcal C}_L, {\mathcal C}_{L+1}, \ldots , {\mathcal C}_{L+[(n+1)/2]-1}$. Suppose $n$ is even. By Lemma \ref{construction2} $B_L=0$ for $1 \leq L \leq n-s-[n/2]+1$. Since $B_1=0$ we have $\sum_{i=1}^{[n/2]}\lambda_i {\mathcal C}_i=0$ for some $\lambda_i \in \mathbb{C}$. We can assume that $\lambda_1\neq 0$. Indeed, by Lemma \ref{construction8}, we can find a non-zero $([n/2]-1) \times ([n/2]-1)$ minor in the block with columns ${\mathcal C}_2, \ldots , {\mathcal C}_{[n/2]}$. We can construct the linear dependence  $\sum_{i=1}^{[n/2]}\lambda_i {\mathcal C}_i=0$ using appropriate $([n/2]-1) \times ([n/2]-1)$ minors as coefficients $\lambda_i$ so we can take $\lambda_1\neq 0$. Thus ${\mathcal C}_1 \in \langle {\mathcal C}_2, {\mathcal C}_3, \ldots , {\mathcal C}_{[n/2]} \rangle$. By repeated application of Lemma \ref{construction8} we can go on to deduce that for $1 \leq j \leq n-s-[n/2]+1$, ${\mathcal  C}_j \in \langle {\mathcal C}_{n-s-[n/2]+2}, \ldots , {\mathcal C}_{n-s}\rangle$. This means that any $[n/2] \times [n/2]$ minor taken from the block $\mathcal B$ equals zero.
So the rank of the matrix (\ref{2m+2n-2rmatrix}) is at most $n-1$.

Let now $n$ be odd. Any non-zero $[(n+1)/2]\times [(n+1)/2]$ minor of $\mathcal B$ must contain the first column and the last row of $\mathcal B$. It is equal by absolute value to the corresponding $[(n-1)/2]\times [(n-1)/2]$ minor of $\mathcal B$ where these row and column are removed. By Lemmas \ref{construction2},   \ref{construction8} applied repeatedly for $2\le L\le [(n+1)/2]-s+1$ we conclude that
the dimension of the space spanned by the columns of ${\mathcal B}$ is at most $(n-1)/2$.

We will now show that the rank of the matrix (\ref{2m+2n-2rmatrix}) is precisely $n-1$. For $1 \leq k \leq [n/2]$, $2 \leq L \leq [(n+1)/2]-s+2$ let $D_L^k$ be the minors formed by taking the square submatrix with columns ${\mathcal C}_L, {\mathcal C}_{L+1}, \ldots , {\mathcal C}_{L+[(n+1)/2]-2}$ where we include all but the $k$th row of ${\mathcal B}$.  Let $n$ be even. Then by Lemma \ref{construction8} $\exists k$ such that $D_L^k\neq 0$. Denote the corresponding matrix by $D_{[\frac{n-1}{2}] \times [\frac{n-1}{2}]}$. Then the following $(n-1) \times (n-1)$ minor of (\ref{2m+2n-2rmatrix})
\begin{equation*}
{\small
\left| \begin{array}{cccccc|c}
\alpha_1^{n-1} & \alpha_1^{n-3} & \ldots & \alpha_1^5 & \alpha_1^3 & \alpha_1&\multirow{5}{*}{0}\\
\alpha_2^{n-1} & \alpha_2^{n-3} & \ldots & \alpha_2^5 & \alpha_2^3 & \alpha_2&\\
\alpha_3^{n-1} & \alpha_3^{n-3} & \ldots & \alpha_3^5 & \alpha_3^3 & \alpha_3&\\
\vdots&\vdots&\ldots&\vdots&\vdots&\vdots \\
\alpha_{[n/2]}^{n-1} & \alpha_{[n/2]}^{n-3} & \ldots & \alpha_{[n/2]}^5 & \alpha_{[n/2]}^3 & \alpha_{[n/2]}&\\
\hline
&&0&&&&D_{[\frac{n-1}{2}] \times [\frac{n-1}{2}]}
\end{array}\right| \neq 0.
}
\end{equation*}
If $n$ is odd then the matrix $D_{[(n-1)/2]\times [(n-1)/2]}$ should be extended by adjoining the column ${\mathcal C}_1$ and the last row of $\mathcal B$.
\end{proof}

\begin{corollary}\label{pm1nexp3} Let $\mathcal{A}=\arrange$. Then
\begin{align*}&P_{even}^{2m+n, 2m+2n-2}
=\frac{(m+2)t^{2m+n}-mt^{2m+n+2}-(m+n+2)t^{2m+2n}+(m+n)t^{2m+2n+2}}{(t^2-1)^2}
\end{align*}
if $n$ is even, and
\begin{align*}&P_{even}^{2m+n, 2m+2n-2}
=\frac{(m+3)t^{2m+n+1}-(m+1)t^{2m+n+3}-(m+n+2)t^{2m+2n}+(m+n)t^{2m+2n+2}}{(t^2-1)^2}
\end{align*}
if $n$ is odd.
\end{corollary}

\begin{lemma}\label{arrangecase2} Let $\mathcal{A}$ satisfy the symmetry property \eqref{assumption}. Let $i=2(m+n-s)$, where $[n/2]+1\le s \le \min(n, m+[(n+1)/2])$. Then $b_i=i/2-[n/2]+1$ in the Hilbert  series (\ref{poincare2}).
\end{lemma}
\begin{proof} The matrix M of the system of linear equations for the non-zero coefficients of a quasi-invariant  $q$ of degree $i$ is equivalent to the form (\ref{2m+2n-2rmatrix}). It is easy to see that $\rank M \le [n/2]+n-s$ as $0\le n-s \le [(n+1)/2]$. For even $n$  the block $\mathcal B$ contains a non-zero $(n-s) \times (n-s)$ minor $B_{(n-s) \times (n-s)}$ by Lemma \ref{construction8}. Hence the $([n/2]+n-s) \times ([n/2]+n-s)$ minor
\begin{equation*}
{\small
\left| \begin{array}{cccccc|c}
\alpha_1^{n-1} & \alpha_1^{n-3} & \ldots & \alpha_1^5 & \alpha_1^3 & \alpha_1&\multirow{5}{*}{0}\\
\alpha_2^{n-1} & \alpha_2^{n-3} & \ldots & \alpha_2^5 & \alpha_2^3 & \alpha_2&\\
\alpha_3^{n-1} & \alpha_3^{n-3} & \ldots & \alpha_3^5 & \alpha_3^3 & \alpha_3&\\
\vdots&\vdots&\ldots&\vdots&\vdots&\vdots \\
\alpha_{[n/2]}^{n-1} & \alpha_{[n/2]}^{n-3} & \ldots & \alpha_{[n/2]}^5 & \alpha_{[n/2]}^3 & \alpha_{[n/2]}&\\
\hline
&&\multirow{2}{*}{0}&&&&B_{(n-s) \times (n-s)}
\end{array}\right|
}\ne 0.
\end{equation*}
The case of odd $n$ is similar.
\end{proof}
As a corollary we get a part of the Hilbert series with even terms preceeding those found in Corollary \ref{pm1nexp3} .

\begin{corollary}\label{pm1nexp4} Let $\mathcal{A}$ satisfy the symmetry property \eqref{assumption}.  If $n \le 2m+1$ then
\begin{align*}&P_{even}^{2m+2, 2m+n-2}
=\frac{(m-n/2+2)t^{2m+2}-(m-n/2+1)t^{2m+4}-(m+1)t^{2m+n}+mt^{2m+n+2}}{(t^2-1)^2}
\end{align*}
if $n$ is even and
\begin{align*}&P_{even}^{2m+2, 2m+n-1}
=\frac{(m-\frac{n-5}{2})t^{2m+2}-(m-\frac{n-3}{2})t^{2m+4}-(m+2)t^{2m+n+1}+(m+1)t^{2m+n+3}}{(t^2-1)^2}
\end{align*}
if $n$ is odd. If $n \ge 2m+1$ then
\begin{align*}&P_{even}^{n+1, 2m+n-2}=\frac{2t^{n+2}-t^{n+4}-(m+1)t^{2m+n}+mt^{2m+n+2}}{(t^2-1)^2}
\end{align*}
if $n$ is even and
\begin{align*}&P_{even}^{n+1, 2m+n-1}=\frac{2t^{n+1}-t^{n+3}-(m+2)t^{2m+n+1}+(m+1)t^{2m+n+3}}{(t^2-1)^2}
\end{align*}
if $n$ is odd.
\end{corollary}

It remains to find the odd coefficients $b_i$ of the Hilbert series $P_{\amn}$  with $\max(2m+1, n+1) \leq i \leq 2m+n-3$ and the coefficients $b_i$ with $n+1 \leq i \leq 2m$. We deal with the latter case first in the following lemma.

\begin{lemma}\label{m1ncase1} Let $\mathcal{A}$ satisfy the symmetry property \eqref{assumption}. Let $ n \leq i \leq 2m$. Then in (\ref{poincare2}) $b_i=\frac{i+1}{2}-[\frac{n+1}{2}]$ if $i$ is odd and $b_i=\frac{i}{2}+1-[\frac{n}{2}]$ if $i$ is even.
\end{lemma}
\begin{proof} Let $q$ be a homogeneous quasi-invariant of degree $i$. It has has no odd powers of $y$ since $i \leq 2m$. Suppose first that $i$ is odd. The matrix of linear equations on the non-zero coefficients of $q$, which express quasi-invariance conditions, is equivalent to
\begin{equation*}
A=\left( \begin{array}{cccc}
\alpha_1^{i-1}& \ldots& \alpha_1^2&1\\
\alpha_2^{i-1}& \ldots& \alpha_2^2&1\\
\vdots&&\vdots&\vdots \\
\alpha_{[\frac{n+1}{2}]}^{i-1}& \ldots& \alpha_{[\frac{n+1}{2}]}^2&1\\
\end{array}\right).
\end{equation*}
Thus $\rank A=\min(\frac{i+1}{2}, [\frac{n+1}{2}])=[\frac{n+1}{2}]$. So the dimension of homogeneous quasi-invariants of degree $i$ is $b_i=i+1-\frac{i+1}{2}-\rank A=\frac{i+1}{2}-[\frac{n+1}{2}]$.
Now let $i$ be even. The matrix expressing the quasi-invariance conditions is equivalent to
\begin{equation*}
\tilde{A}=\left( \begin{array}{cccc}
\alpha_1^{i-1}& \ldots& \alpha_1^3&\alpha_1\\
\alpha_2^{i-1}& \ldots& \alpha_2^3&\alpha_2\\
\vdots&&\vdots&\vdots \\
\alpha_{[\frac{n}{2}]}^{i-1}& \ldots& \alpha_{[\frac{n}{2}]}^3&\alpha_{[\frac{n}{2}]}\\
\end{array}\right).
\end{equation*}
Thus $\rank \tilde{A}=\min(\frac{i}{2}, [\frac{n}{2}])=[\frac{n}{2}]$. So the dimension of homogeneous quasi-invariants of degree $i$ is $b_i=\frac{i}{2}+1-[\frac{n}{2}]$.
\end{proof}

\begin{corollary}\label{pm1nexp1} Let $\mathcal{A}$ satisfy the symmetry property \eqref{assumption}. Suppose that $n \le 2m$. Then
\begin{align*}P^{n+1, 2m}&=\frac{t^{n+1}+2t^{n+2}-t^{n+4}-(m-n/2+1)t^{2m+1}}{(t^2-1)^2}\\
&+\frac{(m-n/2)t^{2m+3}-(m-n/2+2)t^{2m+2}+(m-n/2+1)t^{2m+4}}{(t^2-1)^2}
\end{align*}
if $n$ is even, and
\begin{align*}P^{n+1, 2m}&=\frac{2t^{n+1}+t^{n+2}-t^{n+3}-(m-\frac{n+1}{2}+1)t^{2m+1}}{(t^2-1)^2}\\
&+\frac{(m-\frac{n+1}{2})t^{2m+3}-(m-\frac{n-5}{2})t^{2m+2}+(m-\frac{n-3}{2})t^{2m+4}}{(t^2-1)^2}
\end{align*}
if $n$ is odd.
\end{corollary}

Finally we find the remaining odd coefficients $b_i$ of the Hilbert series $P_{\amn}$.
\begin{lemma}\label{m1ncase2} Let $\mathcal{A}$ satisfy the symmetry property \eqref{assumption}. Let $i$ be odd such that  $2m+n-1\ge i \ge \max(2m-1, n-1)$. Then $b_i=\frac{i+1}{2}-[\frac{n+1}{2}]$ in the Hilbert series (\ref{poincare2}).
\end{lemma}
\begin{proof} Let $q$ be a homogeneous quasi-invariant of degree $i$ where $2m+1 \leq i \leq 2m+n-1$ where $i$ is odd. The matrix expressing quasi-invariance conditions as equations on the coefficients of $q$ can be rearranged to the form
\begin{equation*}
D=\left( \begin{array}{cccc}
0&A\\
B&0
\end{array}\right)
\end{equation*}
where the block $A$ consists of columns $A_j$,  $1 \leq j \leq \frac{i+1}{2}-m$, while the block $B$ consists of  columns $B_j$, $1 \leq j \leq \frac{i+1}{2}$, with

\begin{equation*}
A_j=\left( \begin{array}{c}
\alpha_1^{2j-1}\\
\alpha_2^{2j-1}\\
\alpha_3^{2j-1}\\
\vdots \\
\alpha_{[n/2]}\\
\end{array}\right), \quad
B_j=\left( \begin{array}{c}
\alpha_1^{2j-2}\\
\alpha_2^{2j-2}\\
\alpha_3^{2j-2}\\
\vdots \\
\alpha_{[\frac{n+1}{2}]}^{2j-2}\\
\end{array}\right).
\end{equation*}
Thus $\rank D=\min ([\frac{n+1}{2}], \frac{i+1}{2})+\min(\frac{i+1}{2}-m, [n/2])=[\frac{n+1}{2}]+\frac{i+1}{2}-m$. So we have $b_i=\frac{i+1}{2}-[\frac{n+1}{2}]$.
\end{proof}

\begin{corollary}\label{pm1nexp2} Let $\mathcal{A}$ satisfy the symmetry property \eqref{assumption}.  If $2m \ge n$ then
\begin{align*}&P_{odd}^{2m+1, 2m+n-1}=\frac{(m-n/2+1)t^{2m+1}+(n/2-m)t^{2m+3}-(m+1)t^{2m+n+1}+mt^{2m+n+3}}{(t^2-1)^2}
\end{align*}
if $n$ is even and
\begin{align*}&P_{odd}^{2m+1, 2m+n-2}
=\frac{(m-\frac{n+1}{2}+1)t^{2m+1}+(\frac{n+1}{2}-m)t^{2m+3}-mt^{2m+n}+(m-1)t^{2m+n+2}}{(t^2-1)^2}
\end{align*}
if $n$ is odd. If $2m \le n$ then
\begin{align*}&P_{odd}^{n+1, 2m+n-1}=\frac{t^{n+1}-(m+1)t^{2m+n+1}+mt^{2m+n+3}}{(t^2-1)^2}
\end{align*}
if $n$ is even and
\begin{align*}&P_{odd}^{n+1, 2m+n-2}=\frac{t^{n+2}-mt^{2m+n}+(m-1)t^{2m+n+2}}{(t^2-1)^2}
\end{align*}
if $n$ is odd.
\end{corollary}

Thus we arrive at the main result of this Section.

\begin{theorem}\label{anm-is-gor} The Hilbert series of the  algebra of quasi-invariants $Q_{\amn}$ is given by
\begin{equation*}P_{\amn}(t)=\frac{1-t^2+t^{n+1}+t^{n+2}+t^{2m+n}+t^{2m+n+1}-t^{2m+2n}+t^{2m+2n+2}}{(t^2-1)^2}.
\end{equation*}
\end{theorem}
The Theorem follows from Corollaries \ref{pg1-cor}, \ref{pg2}, \ref{pg3}, \ref{pm1nexp3}, \ref{pm1nexp4}, \ref{pm1nexp1}, \ref{pm1nexp2}.

\section{Gorenstein configurations of type $(m,1^n)$} \label{all-gorenstein-am1n}

Let $\mathcal A$ be a configuration of vectors $\beta_j \in \C^2$, $0\le j \le n$, with multiplicities $m_j \in \N$. As before we assume that $(\beta_j, \beta_j)\ne 0$ $\forall j$. Let $Q_{\mathcal A}\subset \C[x_1,x_2]$ be the associated algebra of quasi-invariants. Let $P_{\mathcal A}(t)$ be its Hilbert series.

\begin{definition}
A configuration $\mathcal A$ is called {\it Gorenstein} if $P_{\mathcal A}(t^{-1})=t^M P_{\mathcal A}(t)$ for some $M\in \Z$.
\end{definition}

This terminology is justified by the fact that the algebra $Q_{\mathcal A}$ is Gorenstein if and only if the configuration $\mathcal A$ is Gorenstein. This follows from the Stanley criterion \cite{Stan} and the following proposition.
\begin{proposition}
The graded ring $Q_{\mathcal A}$ is  Cohen-Macaulay.
\end{proposition}
\begin{proof}
Consider $P_1=x_1^2+x_2^2$, $P_2=\prod_{j=0}^n (\beta_j,x)^{2 m_j}$. It is easy to see that $P_1, P_2 \in Q_{\mathcal A}$. Let $I\subset \C[x_1, x_2]$ be the ideal generated by $P_1, P_2$. Since $P_1, P_2$ have no common zeroes outside the origin it follows that $\C[x_1, x_2]$ and hence $Q_{\mathcal A}$ are finite over $\C[P_1, P_2]$. It is easy to see that  $Q_{\mathcal A}$ is free over $\C[P_1, P_2]$ (see e.g. \cite{DJ} for details).
\end{proof}

Theorem \ref{anm-is-gor} implies that the configuration $\amn$ is Gorenstein. The main result of this Section is the converse statement that there are no other Gorenstein configurations of type $(m,1^n)$.
We establish it  by studying the quasi-invariant conditions and possible Hilbert series $P_{\mathcal A}(t)$.

Let $\mathcal A$ be of type $(m,1^n)$ and recall the notations from the beginning of Section~\ref{quasi-preliminary}. In particular, we fix $\beta_0=(0,1)$ and we denote $\beta_j=(1,\alpha_j)$ for $1\le j \le n$. It will be convenient to introduce the parameter $r$ associated with the configuration $\mathcal A$ as follows:
\begin{equation*}r=\text{number of different}\,\, \alpha_i^2, \,\, i=1,2, \ldots ,n.
\end{equation*}
We assume that the vectors $\beta_j$ are numerated so that $\alpha_1^2, \ldots, \alpha_r^2$ are pairwise different.
Parts of the Hilbert series \eqref{poincare2} depend on the parameter $r$ only rather than on the full geometry of the configuration $\mathcal A$. In this situation we will be using notations
\begin{equation*} P_{r}^{k,l}=\sum_{i=k}^lb_it^i, \quad  P_{r, odd }^{k,l}=\sum_{k\leq 2i+1 \leq l}b_{2i+1}t^{2i+1}, \quad P_{r, even}^{k,l}=\sum_{k \leq 2i \leq l}b_{2i}t^{2i}.
 \end{equation*}

\begin{lemma}\label{pg2m-1<n+11}
If $2r\le m+n$ then
\begin{align*}P_{r, odd}^{n+1, 2m+n-1}=\frac{t^{2r+1}+t^{2n+2m-2r+1}-(m+2)t^{2m+n+1}+mt^{2m+n+3}}{(t^2-1)^2}
\end{align*}
if $n$ is even and
\begin{align*}P_{r, odd}^{n+2, 2m+n-2}=\frac{t^{2r+1}+t^{2n+2m-2r+1}-(m+1)t^{2m+n}+(m-1)t^{2m+n+2}}{(t^2-1)^2}
\end{align*}
if $n$ is odd.

For $2r>m+n$ let $D$ be the matrix expressing the quasi-invariance conditions w.r.t. the vectors $\beta_1,\ldots,\beta_n$ for a homogeneous polynomial $q$ of odd degree $i$ with $n+1 \leq i \leq 2m+n-1$. Suppose that $D$ has maximal possible rank, that is $\emph{\rank} D=\min(i+1-m, n)$ if
\begin{equation*}\frac{i+1}{2}\leq r\,\,\,\,\, \text{and}\,\, \,\,\,n-r<\frac{i+1}{2}-m.
\end{equation*}
If $n, m$ are even then
\begin{align*}P_{odd}^{n+1, 2m+n-1}=\frac{2t^{m+n+1}-(m+2)t^{2m+n+1}+mt^{2m+n+3}}{(t^2-1)^2}.
\end{align*}
If $n, m$ are odd then
\begin{align*}P_{odd}^{n+1, 2m+n-1}=\frac{2t^{m+n+1}-(m+1)t^{2m+n}+(m-1)t^{2m+n+2}}{(t^2-1)^2}.
\end{align*}
If $n$ is even, $m$ is odd then
\begin{align*}P_{odd}^{n+1, 2m+n-1}=\frac{t^{m+n}+t^{m+n+2}-(m+2)t^{2m+n+1}+mt^{2m+n+3}}{(t^2-1)^2}.
\end{align*}
If $n$ is odd, $m$ is even then
\begin{align*}P_{odd}^{n+1, 2m+n-1}=\frac{t^{m+n}+t^{m+n+2}-(m+1)t^{2m+n}+(m-1)t^{2m+n+2}}{(t^2-1)^2}.
\end{align*}
\end{lemma}
\begin{proof} First let $q$ be a homogeneous quasi-invariant of degree $i$ with $n+1 \leq i \leq 2m-1$, $i$ odd. The matrix $D$ expressing the  quasi-invariance conditions as equations on the coefficients of $q$ consists of $\frac{i+1}{2}$ columns $B_j$
\begin{equation*}
B_j=\left( \begin{array}{c}
\alpha_1^{2j-2}\\
\alpha_2^{2j-2}\\
\alpha_3^{2j-2}\\
\vdots \\
\alpha_{r}^{2j-2}\\
\end{array}\right),
\end{equation*}
where $1 \leq j \leq \frac{i+1}{2}$. We have $\rank D=\min(r, \frac{i+1}{2})$, so that
 \begin{equation} \label{inter-n-m}
 b_i=\begin{cases}
0& \text{if}\,\, n+1 \leq i\leq 2r-1, \\
\frac{i+1}{2}-r & \text{if}\,\,  2r+1 \leq i \leq 2m-1.
\end{cases}
\end{equation}

Now let $q$ be a homogeneous quasi-invariant of degree $i$ with $\max(n+1, 2m+1) \leq i \leq 2m+n-1$, and $i$ is odd. The matrix expressing quasi-invariance conditions as equations on the coefficients of $q$ can be rearranged to the form
\begin{equation*}
D=\left( \begin{array}{cccc}
0&A\\
B&*
\end{array}\right),
\end{equation*}
where the block $A$ consists of columns $A_j$,  $1 \leq j \leq \frac{i+1}{2}-m$, while the block $B$ consists of  columns $B_j$, $1 \leq j \leq \frac{i+1}{2}$, with
\begin{equation*}
A_j=\left( \begin{array}{c}
\alpha_{r+1}^{2j-1}\\
\alpha_{r+2}^{2j-1}\\
\vdots \\
\alpha_{n}^{2j-1}\\
\end{array}\right), \quad
B_j=\left( \begin{array}{c}
\alpha_1^{2j-2}\\
\alpha_2^{2j-2}\\
\vdots \\
\alpha_{r}^{2j-2}\\
\end{array}\right).
\end{equation*}
  Suppose first that $2r\le m+n$. Then there are three possibilities for the shapes of the blocks $A, B$:
\begin{enumerate}[(I)]\item $\frac{i+1}{2}\leq r,\quad \frac{i+1}{2}-m\leq n-r,$
\item $\frac{i+1}{2}> r,\quad \frac{i+1}{2}-m \leq n-r,$
\item $\frac{i+1}{2}>r,\quad \frac{i+1}{2}-m > n-r.$
\end{enumerate}
Note that $\rank M=i+1-m$ in the case (I), $\rank M=r+\frac{i+1}{2}-m$ in the case (II) and $\rank M=n$ in the case (III). Hence the dimension $b_i$ of homogeneous quasi-invariants of degree $i$ where $n+1 \leq i \leq 2m+n-1$, with $i$ odd, is given by
\begin{equation} \label{inter-2} b_i=\begin{cases}
0& \text{if}\,\, i\leq 2r-1, \\
\frac{i+1}{2}-r & \text{if}\,\,  2r+1 \leq i \leq 2m+2n-2r-1, \\
i+1-m-n & \text{if}\,\,  2m+2n-2r+1 \leq i \leq 2m+n-1. \\
\end{cases}
\end{equation}
Thus for even $n$ the dimensions \eqref{inter-2} together with \eqref{inter-n-m}   give
\begin{align*}P_{r,odd}^{n+1, 2m+n-1}&=\sum_{\substack{i=2r+1\\i \, odd}}^{2m+2n-2r-1}(\frac{i+1}{2}-r)t^i+\sum_{\substack{i=2n+2m-2r+1\\i \, odd}}^{2m+n-1}(i+1-m-n)t^i\\
&=\sum_{s=r}^{m+n-r-1} (s+1-r) t^{2s+1} + \sum_{s=n+m-r}^{m+n/2-1}(2s+2-m-n)t^{2s+1}\\
&=\frac{t^{2r+1}+t^{2n+2m-2r+1}-(m+2)t^{2m+n+1}+mt^{2m+n+3}}{(t^2-1)^2},
\end{align*}
where we used the identities
\begin{equation}\label{sum-geom}
\sum_{s=a}^b t^s = \frac{t^a-t^{b+1}}{1-t}, \quad \sum_{s=a}^b (s+1)t^s = \frac{(a+1)t^a-a t^{a+1}-(b+2)t^{b+1}+(b+1)t^{b+2}}{(1-t)^2}.
\end{equation}
The case of odd $n$ is similar.

Suppose now that $2r> m+n$.  Then there are three possibilities for the shapes of the blocks $A, B$: the cases (I), (III) are as above while the case (II) is replaced with
\begin{enumerate}[(II')]\item
$\frac{i+1}{2}\leq r,\quad \frac{i+1}{2}-m > n-r$.
\end{enumerate}
We have   $\rank D = \min(i+1-m, n)$ in the case (II') due to the assumptions. Hence the dimension $b_i$ of homogeneous quasi-invariants of  odd degree $i$, where $n+1 \leq i \leq 2m+n-1$, is given by
\begin{equation*}b_i=\begin{cases}
0& \text{if}\,\, i\leq m+n-1, \\
i+1-m-n & \text{if}\,\,  m+n \leq i \leq 2m+n-1,
\end{cases}
\end{equation*}
where we used that for $i\le m+n-1 <2r-1$ we are in the cases (I), (II') with $\rank D = i+1-m$, and that the case (I) is impossible for $i \ge m+n$.
Thus we have
\begin{align*}P_{r,odd}^{n+1, 2m+n-1}&=\sum_{\substack{i=m+n\\i \, odd}}^{2m+n-1}(i+1-m-n)t^i,
\end{align*}
which gives the required expressions with the help of formulas  \eqref{sum-geom}.
\end{proof}

\begin{lemma}\label{pg2m-1<n+12}
If $m+n$ is odd then
\begin{align*}&P_{even}^{n+1, 2m+2n-2}\\
&=\frac{t^{2[\frac{n+2}{2}]}-t^{2[\frac{n+4}{2}]}+t^{m+n+1}+t^{m+n+3}-(m+n+1)t^{2m+2n}+(m+n-1)t^{2m+2n+2}}{(t^2-1)^2} \\
&+\sum_{j=[(n+2)/2]}^{m+n-1} a_{2j}t^{2j},
\end{align*}
where $a_{2j} \in \Z_{\ge 0}$ $\forall j$. If $m+n$ is even then
\begin{align*}&P_{even}^{n+1, 2m+2n-2}\\
&=\frac{t^{2[\frac{n+2}{2}]}-t^{2[\frac{n+4}{2}]}+ 2 t^{m+n+2}-(m+n+1)t^{2m+2n}+(m+n-1)t^{2m+2n+2}}{(t^2-1)^2}
+\sum_{j=[(n+2)/2]}^{m+n-1} a_{2j}t^{2j},
\end{align*}
where $a_{2j} \in \Z_{\ge 0}$ $\forall j$.
\end{lemma}
\begin{proof}
Let $b_i$ be the dimension of the space of homogeneous quasi-invariant polynomials of even degree $i$ such that $n+1 \le i \le 2(m+n-1)$.
We claim that $b_i \ge 1$ for $i\le m+n$ and $b_i \ge i+1-m-n$ for $i>m+n$. Indeed, consider the matrix $D$ expressing the quasi-invariant conditions for the non-zero coefficients of the polynomial of degree $i$. It is easy to see that the columns of $D$ are linearly dependent hence $b_i \ge 1$. Let now  $i>m+n$. Assume firstly that $i\ge 2m$. Then $D$ has $i+1-m>n$ columns and $\rank D \le n$, hence $b_i \ge i+1-m-n$ as required. Suppose now that $m+n<i<2m$. Then after the elementary transformations the matrix $D$ takes the form

\begin{equation*}
\tilde{D}=\left( \begin{array}{cccc}
\alpha_1^{i-1}& \ldots& \alpha_1^3&\alpha_1\\
\alpha_2^{i-1}& \ldots& \alpha_2^3&\alpha_2\\
\vdots&&\vdots&\vdots \\
\alpha_{r}^{i-1}& \ldots& \alpha_{r}^3&\alpha_{r}\\
\end{array}\right),
\end{equation*}
so $\rank D = \rank \tilde D \le \min(i/2, r)\le \min(m,n)=n$. Then $b_i \ge i/2+1-n > i+1-m-n$ since $i<2m$.

Let $P_{g, even}^{n+2, 2m+2n-2}$ be the segment of the Hilbert series where we take the minimal values $b_i$ from the above estimates:
\begin{equation}\label{pg}
P_{g, even}^{n+1, 2m+2n-2}
=\sum_{\nad{i=n+1}{ i \,\,  even}}^{m+n}t^i+\sum_{\nad{i=m+n+1}{i \,\, even}}^{2m+2n-2}(i+1-m-n)t^i.
\end{equation}
If $m+n$ is odd then \eqref{pg} takes the form
\begin{align}\label{calc1gen}&P_{g, even}^{n+1, 2m+2n-2}
=\sum_{i=[(n+2)/2]}^{(m+n-1)/2}t^{2i}+\sum_{i=(m+n+1)/2}^{m+n-1}(i+1-m-n)t^{2i}  \notag  \\
&=\frac{t^{2[(n+2)/2]}-t^{2[(n+4)/2]}+t^{m+n+1}+t^{m+n+3}-(m+n+1)t^{2m+2n}+(m+n-1)t^{2m+2n+2}}{(t^2-1)^2}
\end{align}
upon applying \eqref{sum-geom}. Similarly, if $m+n$ is even then
\begin{align}\label{calc2gen}&P_{g, even}^{n+1, 2m+2n-2}
=\sum_{i=[(n+2)/2]}^{(m+n)/2}t^{2i}+\sum_{i=(m+n+2)/2}^{m+n-1}(i+1-m-n)t^{2i}
\end{align}
and the statement follows from the relations \eqref{sum-geom}.
\end{proof}

Corollaries \ref{pg1-cor},
 \ref{pg2}, \ref{pg3} and Lemmas  \ref{pg2m-1<n+11}, \ref{pg2m-1<n+12} imply the following statement.

\begin{corollary}\label{pg2m-1<n+14}  If $2r\le m+n$ then the Hilbert series \eqref{poincare2} takes the form
\begin{equation}\label{possible-hilb}
P(t)=\frac{P_{g,r}(t)}{(t^2-1)^2} + \sum_{j=[(n+2)/2]}^{m+n-1} a_{2j} t^{2j},
\end{equation}
where $a_{2j}\in \Z_{\ge 0} \, \forall j$ and
$$
P_{g,r}(t) = 1-t^2+t^{2r+1}+t^{2n+2m-2r+1}+t^{m+n+1}+t^{m+n+3}
$$
when $m+n$ is odd, and
$$
P_{g,r}(t) = 1-t^2+t^{2r+1}+t^{2n+2m-2r+1}+2 t^{m+n+2}
$$
when $m+n$ is even.

If $2r>m+n$ then the Hilbert series \eqref{poincare2} takes the form
$$
P(t)=\frac{P_{g}(t)}{(t^2-1)^2} + \sum_{j=[(n+2)/2]}^{m+n-1} a_{2j} t^{2j} + \sum_{i=[\frac{n+1}{2}]}^{m+[n/2]-2}a_{2i+1}t^{2i+1},
$$
where $a_{2j}, a_{2i+1}\in \Z_{\ge 0} \, \forall i, j$ and
$$
P_{g}(t) = 1-t^2+t^{m+n}+t^{m+n+1}+t^{m+n+2}+t^{m+n+3}
$$
when $m+n$ is odd, and
$$
P_{g}(t) = 1-t^2+2t^{m+n+1}+2t^{m+n+2}
$$
when $m+n$ is even.
\end{corollary}

\begin{remark}\label{remark-on-a-critical}
Note that it follows from Proposition \ref{lemma2m+2n-2}, Lemmas \ref{arrangecase1}, \ref{pg2m-1<n+12} and the calculations \eqref{pg}-\eqref{calc2gen} that $a_{2(m+n-1)}\ne 0$ if and only if $\mathcal A$ is equivalent to the configuration $\amn$, in which case $a_{2(m+n-1)}=1$.
\end{remark}

We are now ready to prove the main result of this Section.

\begin{theorem}\label{gorenimpliesarrangeallthecases} Let a configuration $\mathcal A$ of type $(m,1^n)$ be Gorenstein. Then $\mathcal{A}$ is equivalent to the configuration $\arrange$.
\end{theorem}

\begin{proof} We have $P(t)$ is palindromic. Suppose first that $2r\le m+n$.
 We rearrange the terms given by \eqref{possible-hilb} to the common denominator. Observe that the only terms of odd degree in the numerator are $t^{2r+1}$ and $t^{2(m+n-r)+1}$. Now, suppose the degree of the numerator is not $2(m+n+1)$. Then $t^{2r+1}$ and $t^{2(m+n-r)+1}$ cannot `match', so the total degree must be odd. However this means $-t^2$ must match with some term with odd power, which is not possible as the coefficients are different. So the total degree is $2(m+n+1)$. This means that $a_{2(m+n-1)}\ne 0$. It follows from Remark \ref{remark-on-a-critical} that $\mathcal{A}$ is equivalent to $\arrange$.

Now suppose that $2r > m+n$. Note that in this case $r \neq [\frac{n+1}{2}]$ and so $\mathcal{A} \neq \arrange$. By Remark \ref{remark-on-a-critical} this implies $a_{2m+2n-2}=0$.  Suppose that $n+m$ is even and that $\mathcal A$ is Gorenstein. Let us rearrange the series $P(t)$ as
{\small
\begin{align} \label{m+n-even}
&P(t)= \nonumber \\
&\frac{1-t^2+2t^{m+n+1}+2t^{m+n+2}+(t^2-1)^2\sum_{i=[\frac{n+1}{2}]}^{m+[n/2]-2}a_{2i+1}t^{2i+1} + \sum_{j=[\frac{n}{2}]+1}^{m+n+1}(a_{2j-4}-2a_{2j-2}+a_{2j})t^{2j}}{(t^2-1)^2},  \nonumber \\
\end{align}
}
where 
we put $a_{2[n/2]-2}=a_{2[n/2]}=a_{2(m+n)}=a_{2(m+n+1)}=0$.
Let $d$ be the degree of the numerator of  \eqref{m+n-even}.

Suppose that $d$ is odd. Then $d\le 2m+2[n/2]+1$. Notice that all the even powers $t^{2j}$ in the numerator of  \eqref{m+n-even} vanish for $m+n+2 \le 2j \le 2(m+n)$. This follows from the palindromicity of the numerator of  \eqref{m+n-even}. Indeed, if any of these powers has non-zero coefficient then it should match with some odd power $t^k$ since $d$ is even. However the coefficient at $t^k$ is zero unless $k \ge 2[(n+1)/2]+1$. In the latter case  $k+ 2j \ge 2[(n+1)/2]+1+m+n+2 >d$ since $n\ge m$, as $r \le n$, so this is impossible and the mentioned even powers vanish. Consider the even powers $t^{2j}$ with $m+n+4\le 2j\le 2(m+n)$. It follows recursively that $a_{2j-4}=0$. Consider now the even power $t^{m+n+2}$. It comes with the coefficient  \begin{equation}
\label{a-nonzero}
a_{m+n+2}-2 a_{m+n}+a_{m+n-2}+2 = a_{m+n-2}+2 \ne 0.
\end{equation}
This is a contradiction which implies that $d$ is even.

Consider now the odd terms in the numerator  of \eqref{m+n-even} which form the palindromic polynomial by themselves. Let $p,q\in \N$, $p\le q$,  be such that $a_{2p+1}$ and $a_{2q+1}$ be respectively  the first and the last non-zero coefficient $a_{2i+1}$ from the numerator. Note that they exist since otherwise the only remaining odd term $t^{m+n+1}$ has degree bigger than $d/2$. Notice that the degree $n+m+1 \ge 2p+1$ and define $\widehat q = \max(q, (m+n)/2-2)$.
 Define $c_j=a_{2 \widehat q+1-2j} - a_{2p+1+2j}$. For $0 \le j < \widehat q-\frac{m+n}{2}+2$ we have from the palindromicity that $c_j = 2c_{j-1}-c_{j-2}$ and it follows that $c_j=0$.  Then $c_{\widehat q-(m+n)/2+2}=-2$. Let now $\widehat q-(m+n)/2+2 \le j \le [(\widehat q-p+1)/2]$. It follows from $c_j = 2c_{j-1}-c_{j-2}$ recursively that $c_j=-2j+2 \widehat q-m-n+2$.

If $\widehat q-p$ is even then $c_{(\widehat q-p)/2}=p+\widehat q-m-n+2<0$ since $d=(2p+1)+(2 \widehat q+5)\le 2(m+n)$. On the other hand $c_{(\widehat q-p)/2}=0$ from the definition, which is a contradiction. If $\widehat q-p$ is odd then it follows that $c_{(\widehat q-p-1)/2}+c_{(\widehat q-p+1)/2}=2(p+\widehat q-m-n+2)<0$. On the other hand $c_{(\widehat q-p-1)/2}+c_{(\widehat q-p+1)/2}=0$ from the definition, which is a contradiction. Hence ${\mathcal{A}}$ is not Gorenstein when $2r>m+n$.

In the case $m+n$ is odd the arguments are similar. The odd term $2t^{m+n+1}$  in the numerator of \eqref{m+n-even} is replaced with two odd terms $t^{m+n}+t^{m+n+2}$. The total degree $d$ is still even. Indeed, the arguments are same except that the consideration of even power $t^{m+n+3}$ gives $
a_{m+n+3}-2 a_{m+n+1}+a_{m+n-1}+1 = a_{m+n-1}+1 \ne 0$ in place of
 the relation \eqref{a-nonzero}. Then $\widehat q$ should be defined by $\widehat q = \max (q, (n+m+1)/2-2)$
and one gets $c_{\widehat q-(m+n+1)/2+2}=-1$. The formula  $c_j=-2j+2 \widehat q-m-n+2$ for $j \ge \widehat q-(m+n+1)/2+2$ and the subsequent arguments remain unchanged.
\end{proof}

\section{Concluding remarks}

In this paper we studied the class $\mathcal{BA}$ of Baker-Akhiezer configurations on the plane. In the case when at most one multiplicity is arbitrary our results are most complete. Indeed, all the corresponding configurations  are explicitly described in terms of symmetric polynomials of the coordinates of the vectors, there is a description in terms of Darboux transformations too. In this case we also computed the Hilbert series of the corresponding algebras of quasi-invariants and noted the algebras are Gorenstein. We described the class $\mathcal{G}$ of all configurations with Gorenstein quasi-invariants and arrived to the same configurations:
\begin{equation}\label{bagor}
\mathcal{ BA = G}.
\end{equation}
It would be interesting to clarify whether the coincidence \eqref{bagor} holds for more general multiplicities on the plane and also for the configurations in higher dimensions. Although this looks plausible none of the two inclusions seems clear to us. A Gaussian bilinear form on the space of quasi-invariants when a configuration satisfies the conditions  \eqref{1stconds} can be defined (cf. \cite{FHV}), it might be relevant to the analysis of the Gorenstein property.  Furthermore, it would be important to clarify in the two-dimensional case whether the configurations ${\mathcal A}^q_{(m, \tilde m, 1^n)}$ that we considered in this paper exhaust the set $\mathcal{BA}$.

A further remark is on the class of locus configurations $\mathcal{BA}_w$ that admit the weaker version of the Baker-Akhiezer function \cite{CFV}. In the two-dimensional case these configurations are  described in \cite{BL}, \cite{B} (see also \cite{Muller}).  As all our planar configurations $\mathcal{BA}$ are happen to be real it would be interesting to clarify whether the subclass of real configurations from $\mathcal{BA}_w$ coincides with the class $\mathcal{BA}$ both in the two-dimensional case and in higher dimensions.

Finally, papers \cite{Zh1, Zh2} deal with the correspondence between `geometric data' and commutative rings of partial differential operators, particularly, in two variables. Thus it would be interesting to describe the geometric data corresponding to the rings of quasi-invariants considered in this paper.

\section{Acknowledgements}
We are grateful to O.A. Chalykh for useful discussions.
The work of both authors was supported by EPSRC grant EP/F032889/1.
 MF also acknowledges support  from Royal Society/RFBR joint project JP101196/11-01-92612.


\begin{thebibliography}{10}
\bibitem{CV}
O.A. Chalykh, A.P. Veselov  \emph{Commutative rings of partial differential operators and Lie algebras}, Commun. Math. Phys. \textbf{126} (1990), 597--611.

\bibitem{CSV}
O.A. Chalykh, K.L. Styrkas, A.P. Veselov
\emph{Algebraic integrability for the {S}chr{\"{o}}dinger equation and finite reflection groups}, Theor. Math. Phys. \textbf{94} (1993), 253--275.




\bibitem{CFV}
O.A. Chalykh, M.V. Feigin, A.P. Veselov
\emph{Multidimensional Baker-Akhiezer functions and Huygens' principle}, Commun. Math. Phys. \textbf{206} (1999), 533--566.



\bibitem{BerVes}

Yu.Yu. Berest, A.P. Veselov {\em The Huygens principle and integrability}. (Russian) Uspekhi Mat. Nauk {\bf 49} (1994), no. 6(300), 7--78; translation in Russian Math. Surveys {\bf 49} (1994), no. 6, 5--77.

\bibitem{BL}
Yu.Yu. Berest,  I.M. Lutsenko
\emph{Huygen's principle in Minkowski spaces and soliton solutions of the Korteweg-de Vries equation}, Commun. Math. Phys.
\textbf{190} (1997), 113--132.


\bibitem{CFVdeformednote}
O.A. Chalykh, M.V. Feigin, A.P. Veselov
\emph{New integrable deformations of the quantum {C}alogero-{M}oser problem}, Russ. Math. Surv. \textbf{51} (1996), 185--186.

\bibitem{CFVnote2}
O.A. Chalykh, M.V. Feigin, A.P. Veselov
\emph{New integrable generalizations of {C}alogero-{M}oser quantum problems}, J. Math. Phys.
\textbf{39} (1998), 695--703.




\bibitem{B}
Yu.Yu. Berest {\em Solution of a restricted Hadamard problem on Minkowski spaces}, Comm. Pure Appl. Math. {\bf 50} (1997), no. 10, 1019--1052.

\bibitem{Muller}
G. Muller  \emph{2D Locus Configurations and the Charged Trigonometric Calogero-Moser System.},  J. Nonlinear Math. Phys. {\bf 18} (2011), no. 3, 475--482.

\bibitem{FHV}
M.V. Feigin, M.A.~Halln\"as, A.P. Veselov {\em Baker-Akhiezer functions and Macdonald-Mehta integrals}, arXiv:1210.5270

\bibitem{FV}
M. Feigin, A.P Veselov
\emph{Quasi-invariants of {C}oxeter groups and m-harmonic polynomials}, Internat. {M}ath. {R}es. {N}otices \textbf{10} (2002), 521--545.


\bibitem{eg}
P. Etingof and V. Ginzburg
\emph{On m-quasiinvariants of a {C}oxeter group}, Moscow Math. Journal \textbf{2} (2002), 555--566.

\bibitem{FeV}
G. Felder, A.P. Veselov
\emph{Action of Coxeter groups on m-harmonic polynomials and Knizhnik-Zamolodchikov equations}, Moscow Math. Journal \textbf{3} (2003), no. 4, 1269--1291.

\bibitem{FeiginVes}
M. Feigin, A.P Veselov
\emph{Quasi-invariants and quantum integrals of the deformed {C}alogero-{M}oser systems}, Internat. {M}ath. {R}es. {N}otices \textbf{46} (2003), 2488--2511.


\bibitem{Ch}
O.A. Chalykh {\em A construction of commutative rings of differential operators}, Mathematical Notes, \textbf{53} (1993), Nu. 3, 329--335.


\bibitem{BCE}
Y. Berest,  T. Cramer, F. Eshmatov
\emph{Heat kernel coefficients for two-dimensional Schr{\"{o}}dinger operators}, Commun. Math. Phys.
\textbf{283} (2008), 853--860.

\bibitem{Gessel}
I.M Gessel, D. Stanton
\emph{Short {P}roofs of {S}aalsch{\"{u}}tz's and {D}ixon's Theorems}, J. Comb. Theory,
\textbf{38} (1985), 87--90.

\bibitem{DJ}
D. Johnston
\emph{Quasi-invariants of hyperplane arrangements}, PhD thesis, University of Glasgow (2011).

\bibitem{Stan}
R. Stanley
\emph{Hilbert functions of graded algebras}, Advances in Math., {\bf 28} (1978), no. 1, 57--83.


\bibitem{Zh1}
A.B. Zheglov \emph{On rings of commuting partial differential operators}, arXiv:1106.0765.

\bibitem{Zh2}
H. Kurke, D. Osipov, A. Zheglov \emph{Commuting differential operators and higher-dimensional algebraic varieties}, arXiv:1211.0976.

\end{thebibliography}
\end{document}